\documentclass[12pt,a4paper,english]{article}

\usepackage[utf8]{inputenc}
\usepackage[T1]{fontenc}
\usepackage[english]{babel}

\usepackage[usenames,dvipsnames]{xcolor}
\usepackage{pgf,tikz}
\usetikzlibrary{arrows}
\usetikzlibrary{decorations.markings,decorations.pathmorphing,patterns}
\usepackage{pgfplots}

\usepackage[format=hang,font=small]{caption}
\usepackage{booktabs,bm,multirow,subfig,comment}%
\usepackage{enumitem}

\usepackage[intlimits]{amsmath}
\usepackage{amsfonts,amssymb,amsthm}%
\DeclareMathAlphabet{\mathcalligra}{OT1}{pzc}{m}{it}

\usepackage%[draft]
{graphicx}

\graphicspath{{Images/}}

\renewcommand{\phi}{\varphi}
\renewcommand{\theta}{\vartheta}
\renewcommand{\epsilon}{\varepsilon}
\newcommand{\mychi}{\protect{\raisebox{0.15em}{$\chi$}}}

\newcommand{\field}[1]{\mathbb{#1}} 
\newcommand{\R}{\field{R}}
\newcommand{\Z}{\field{Z}}
\newcommand{\N}{\field{N}}

\newcommand{\CC}{\mathcal{C}}
\newcommand{\EE}{\mathcal{E}}
\newcommand{\DD}{\mathcal{D}}
\newcommand{\QQ}{\mathcal{Q}}
\newcommand{\FF}{\mathcal{F}}
\newcommand{\WW}{\mathcal{W}}
\newcommand{\VV}{\mathcal{V}}
\newcommand{\De}{\DD^{(1)}_\epsilon}
\newcommand{\MM}{\mathcal{M}}
\newcommand{\Rc}{\mathcal{R}}
\DeclareMathOperator{\dist}{dist}

\DeclareMathOperator{\inter}{int}
\DeclareMathOperator{\Imm}{Im}
\DeclareMathOperator{\arccot}{arccot}

\providecommand{\abs}[1]{\left\lvert#1\right\rvert}                             %valore assoluto
\providecommand{\norm}[1]{\left\lVert#1\right\rVert}

\newcommand{\dert}[1]{\dot{#1}} %derivata, anche {\frac{d#1}{dt}}

\newcommand*{\diff}{\mathop{}\!\mathrm{d}}
\newcommand{\dd}{\diff}
\newcommand{\Der}[2]{\frac{\diff #1}{\diff #2}}
 
\newtheorem{theorem}{Theorem}
\newtheorem{lemma}[theorem]{Lemma}
\newtheorem{prop}[theorem]{Proposition}

\theoremstyle{definition}

\numberwithin{equation}{section}
\usepackage[affil-it,auth-sc]{authblk}

\title{On the genesis of directional friction through bristle-like mediating elements}

\author[1,\footnote{Corresponding author -- e-mail address:\,\texttt{pgidoni@sissa.it} -- phone:\,+39\,040\,3787\,591}]{Paolo Gidoni}
\author[1,2]{Antonio DeSimone}

\affil[1]{\small{SISSA - International School for Advanced Studies, via Bonomea 265, 34136 Trieste, Italy.} \smallskip}
\affil[2]{\small{GSSI - Gran Sasso Science Institute, viale Francesco Crispi 7, 67100 L'Aquila, Italy.} \smallskip}
\begin{document}
%%%%%%%%%%%%%%%%%%%%%%%%%%%%%%%%%%%%%%%%%%%%%%%%%%%%%%%%%%%%%%%%%%%%%%%%%%%%%%%%%%%%%%%%%

\maketitle

\begin{abstract}
We propose an explanation of the genesis of directional dry friction, as emergent property of the oscillations produced in a bristle-like mediating element by the interaction with microscale fluctuations on the surface.

Mathematically, we extend a convergence result by Mielke, for Prandtl--Tomlinson-like systems, considering also non-homothetic scalings of a wiggly potential. This allows us to apply the result to some simple mechanical models, that exemplify the interaction of a bristle with a surface having small fluctuations. We find that the resulting friction is the product of two factors: a geometric one, depending on the bristle angle and on the fluctuation profile, and a energetic one, proportional to the normal force exchanged between the bristle-like element and the surface. Finally, we apply our result to discuss the with the nap/against the nap asymmetry.
\end{abstract}

\section{Introduction}

Modelling frictional interactions is a challenging task, both for the variety of behaviours experimentally observed, and for the relevance of such phenomena in the study and control of mechanical devices. 
The common strategy consists of a multiscale approach, where the frictional behaviour is an emergent macroscopic property of  mechanical interactions between the asperities of the two surfaces occurring at microscale \cite{aBS09}. Such interactions are often described by modelling the asperities with simple mechanical systems, such as springs and bristles \cite{DriBer,deW95,HaeFri}.

A classical example of such multiscale approach is the Prandtl-Tomlinson model, developed to explain the genesis of Coulomb dry friction \cite{Pop10,PopGre14}.
The model considers the motion of a point mass along a sinusoidal potential, subject to an external driving force and a viscous damping, showing convergence to a dry friction behaviour when the sinusoidal oscillations decrease homothetically. This scenario can be related to the interaction of a single asperity of the upper surface with a rigid rough lower surface. Such representation applies also to the interaction of the cantilever with the surface in a friction force microscope. 

\medbreak

In this paper we follow this multiscale paradigm to propose an explanation of the genesis of a directional asymmetry in the coefficients of Coulomb dry friction, in situations where the interaction between the two surfaces is mediated by bristle-like elements.

Our work is motivated by a growing interest in the modelling and development of  crawling locomotors, exploiting an asymmetry in the friction coefficients \cite{DeSGidNos15,DeSTat12,GidDeS16,GidNosDeS14,ZimZei07}. Such directionality of frictional forces is common both in Nature \cite{Mil88} and in bio-inspired robots \cite{Mah04,Men06,NosDeS14}, and is usually obtained thanks to elastic elements,  such as oblique filaments or bristles (e.g.~the \emph{setae} in earthworms), that mediate the interaction between the crawler and the surface \cite{Han}.

Our starting point is the paper  \cite{Mie12} by Mielke.  Here, it is shown that the quasi-static behaviour of a family of Prandtl--Tomlinson-like systems, in which the fluctuation in the potential decreases homothetically, converge to that of a particle subject to dry friction. Moreover, the leftwards and rightwards friction coefficients coincide with the minimum and maximum of the derivative of the oscillating potential. In this way, a directionality in the friction is produced simply by assuming a suitable asymmetry in the potential. As we discuss in Section \ref{sec:strat}, a key element in this approach is the change in the nature of the dissipation, from viscous in the approximating systems to rate-independent in the limit one.

\medbreak

We apply these ideas to study the limit behaviour of systems characterized by a mediating, bristle-like element, interacting with a wiggly surface whose periodic fluctuations scale homothetically to zero. In this way the wiggly potential is generated by the small oscillations in the mediating element, induced by the fluctuations of the surface.
Moreover an asymmetry in the wiggly potential can be simply produced by  an asymmetry in the mediating element (e.g.~the inclination of a bristle), also in the case of a symmetric surface.

In order to apply Mielke's approach to our problem, we need to extend his framework to more general families of approximating systems, in which the scaling of the wiggly potential is no longer homothetic, but contains also a nonlinear term (cf.~eq.~\eqref{eq:defV}). This is our first result, presented in Section \ref{sec:abstr} (Theorem \ref{th:main}), and constitutes the abstract contribution of this paper.

\medbreak

From the point of view of applications, our main result is to provide some physical insight into the origin of directional friction. This is obtained by constructing some concrete examples of simple mechanical systems producing, in the limit, directional dry friction, and by interpreting the origin of this frictional asymmetry in terms of the parameters characterizing each example.

The friction coefficients we obtain are the product of two factors. The first one is ``geometric'': it contains the asymmetry of the system and is determined only by the roughness of the surface and by the angle of the mediating bristle-like element. The second factor is   instead ``energetic'': it depends on the limit energetic state of the mediating element, but not on the direction of motion. This last coefficient is proportional to the normal force exerted, at the limit, by the mediating element on the surface. In this way we recover the classical structure of Coulomb friction law, where the friction force is the product of a coefficient characteristic of the surfaces and the modulus of the normal forces exchanged between them.

Our results are then used to discuss the \emph{with the nap/against the nap} asymmetry. 
As we will argue better in Section \ref{sec:contropelo}, our intuition of such asymmetry actually includes under the same name several distinct phenomena, producing the same kind of directionality. Despite the complex behaviour that can be showed by a bristle, our model of Section \ref{sec:3m} can be used to outline two fundamental effects, corresponding to changes in the two factors that characterize the friction coefficients. The \emph{geometric effect} occurs when the bristle keeps the same configuration during the two phases (with and against the nap), and the directionality is due to the inclination of the bristle, that in this way ``perceives'' a symmetric fluctuation of the surface as asymmetric. 
The \emph{energetic effect} applies to situations where the configuration of the bristle flips when the velocity changes sign, so that the tip of the bristle is always behind its root with respect to the direction of motion. In this case the geometric component is unchanged, but the bristle switches between two different energetic states, exerting a different normal force on the surface.

Finally, we notice that the behaviour of the model of Section \ref{sec:3m} has a close resemblance to that observed experimentally for the robotic crawler developed in \cite{NosDeS14}. There, slanted bristles, interacting with a groove-textured surface, are used to obtain net displacement, when the body of the crawler performs a cycle of elongation and contraction.
 The bristle-surface interaction produces an oscillatory friction force, and it is shown that 
the system can be effectively discussed considering supports moving on a flat surface and experiencing a constant average friction force. 
  Such result supports our approach and encourages a future experimental validation of the predictions of our models.  

%%%%%%%%%%%%%%%%%%%%%%5

\section{Abstract setting} \label{sec:abstr}

In this section we show that the evolution of a prototype one dimensional rate independent system, with energy $\EE$ and a  dissipation potential $\Rc$ positively homogeneous of degree $1$, can be constructed 
 as the limit of the evolutions of a family of systems $(\EE_\epsilon,\Rc_\epsilon)$, where $\EE_\epsilon=\EE+V_\epsilon$, with $V_\epsilon$ an oscillatory (``wiggly'') small perturbation, and $\R_\epsilon$ is a small viscous dissipation potential.
The system $(\EE_\epsilon,\Rc_\epsilon)$ will describe a motion on an undulatory surface with vanishing small roughness, while the system $(\EE,\Rc)$ describes motion on a flat surface with directional dry friction.

\medbreak

Let us therefore consider a mechanical system having internal energy
\begin{equation} \label{eq:defE}
\EE(t,z)=\Phi(z)-\ell (t)z 
\end{equation}
where $t\in[0,T]$ represent the time and  $z\in \R$ is a one-dimensional state variable.
We assume that $\Phi\in \CC^2(\R,\R)$ is a uniformly convex function, while $\ell\in \CC^1([0,T],\R)$.
The dissipative effects of a change in the state of the system is described by the dissipation potential
\begin{equation} \label{eq:defRlim}
\Rc(v)=\begin{cases}
\rho_+ v &\text{for $v\geq 0$}\\
\rho_- v &\text{for $v\leq 0$}
\end{cases}
\end{equation}
where $\rho_-<0$ and $\rho_+>0$ are suitable constants.
We consider the quasi-static evolution of the system, described by
\begin{equation}
0\in \partial_{\dert z} \Rc({\dert z})+D_z\EE(t,z)
\label{eq:din_lim}
\end{equation}
where the dot $\dert{}$ denotes the derivative with respect to the time variable $t$, $\partial_{\dert z}$ denotes the subdifferential with respect to $\dert z$ and  $D_z$ denotes the derivative in the $z$ variable (below also denoted briefly, when not ambiguous, with a prime $'$).

\medbreak

Similarly  we introduce the following family of perturbed systems depending on a small parameter $\epsilon$.
The energy of these systems is obtained by adding to $\EE$ a small wiggly perturbation. More precisely we have
\begin{equation} \label{eq:defEeps}
\EE_\epsilon (t,z)=\Phi(z)-\ell (t)z + V_\epsilon(z) 
\end{equation}
with
\begin{equation}
V_\epsilon(z)= \epsilon W\left(\frac{z}{\epsilon}\right)+ \epsilon^2 Q\left(\epsilon;\frac{z}{\epsilon}\right)
\label{eq:defV}
\end{equation}
Here $W\in \CC^2(\R,\R)$ is a $1$-periodic (non-constant) function; whereas $Q\colon (0,\epsilon_Q)\times \R\to \R$, for some $\epsilon_Q>0$, is $1$-periodic and $\CC^2$ in the second variable. Moreover we assume the existence of two  positive constants $C_{Q,0}$ and $C_{Q,1}$ such that, for every $0<\epsilon<\epsilon_Q$ and for every $y\in \R$ we have
\begin{align}
\abs{Q(\epsilon;y)}<C_{Q,0} && \abs{Q'(\epsilon;y)}<C_{Q,1} 
\end{align}
where the prime $'$ denotes the derivative with respect to the second variable~$y$.

\medbreak

The systems are subject to a viscous friction, described by the Rayleigh dissipation potential
\begin{equation}\label{eq:defReps}
\Rc_\epsilon(\dert z)=\frac{\epsilon^\gamma}{2}{\dert z}^2 \qquad \text{for some $\gamma>0$}
\end{equation} 
and their (quasi-static) evolution is described by the equation
\begin{equation}
0=D_{\dert z} \Rc_\epsilon(\dert z)+D_z\EE_\epsilon(t,z)
\label{eq:din_eps}
\end{equation}

We are going to show that the behaviour of the system \eqref{eq:din_lim} is approximated, for $\epsilon\to 0$, by that of the systems \eqref{eq:din_eps}. To do so, a last assumption is needed, in order to link the two situations. Namely, we require
\begin{align}
\rho_+&=\max W'(z)>0 &
\rho_-&=\min  W'(z)<0 
\end{align}
We are now ready to state the main result of this section.

\begin{theorem}\label{th:main}
In the framework described above, let $z_\epsilon \colon [0,T]\to \R$ be a family of solutions of \eqref{eq:din_eps}, such that
\begin{equation}
z_\epsilon(0)\to z^0\in (\Phi')^{-1} \left ( [\ell (0)-\rho_+, \ell (0)-\rho_- ]\right)
\label{eq:cond_zetazero}
\end{equation}
Then, the differential inclusion \eqref{eq:din_lim} has a unique solution $\bar z\colon [0,T]\to \R$ for the initial conditions $\bar z(0)=z^0$. Moreover, for $\epsilon\to 0$, this solution satisfies
\begin{gather}
z_\epsilon\to \bar z \qquad \text{in $\CC^0([0,T])$}\\
\int_{t_1}^{t_2}2\Rc_\epsilon(\dert z_\epsilon(t))\dd t \to \int_{t_1}^{t_2}\Rc(\dert{\bar z}(t))\dd t
\qquad \text{for every $0\leq t_1<t_2\leq T$} \label{eq:th_Rconv}
\end{gather}
\end{theorem}

Theorem \ref{th:main} is proved in Section \ref{sec:proof}, through a convergence strategy illustrated in Section \ref{sec:strat}.
Let us remark that the right term in \eqref{eq:cond_zetazero}  is well defined since, being $\Phi$ uniformly convex, it follows that $\Phi'$ is globally invertible with range equal to $\R$. 

\medbreak

For our application, it is useful to study an apparently more general situation and show that it actually falls in the framework of Theorem \ref{th:main}.
Let us consider a function $\FF\in \CC^3([-\delta_\FF,\delta_\FF],\R)$ defined in a neighbourhood of zero and such that 
\begin{equation} \label{eq:cond_alpha}
\FF'(0)=\alpha\neq 0
\end{equation}
Let   $\WW\in \CC^2(\R,\R)$ be a $1$-periodic (non-constant) function and set
\begin{align}
\mu_+&=\max \WW'(z)>0  &
\mu_-&=\min  \WW'(z)<0 
\end{align}

We consider also a function $\QQ\colon (0,\tilde\epsilon_\QQ)\times \R\to \R$, defined for some $\tilde\epsilon_\QQ>0$, and such that it is $1$-periodic and $\CC^2$ in the second variable. We assume that there exist two  positive constants $\tilde C_{\QQ,0}$ and $\tilde C_{\QQ,1}$ such that, for every $0<\epsilon<\tilde\epsilon_\QQ$ and for every $y\in \R$, we have
\begin{align}
\abs{\QQ(\epsilon;y)}<\tilde C_{\QQ,0} && \abs{\QQ'(\epsilon;y)}<\tilde C_{\QQ,1} 
\end{align}

Let $\epsilon_\FF$ be small enough to satisfy $\epsilon_\FF \norm{\WW}_\infty+ \epsilon_\FF^2 \tilde C_{\QQ,0}<\delta_\FF$. We now consider, for every positive $\epsilon<\min\{\epsilon_\FF,\epsilon_\QQ\}$, the general wiggly potential $\VV_\epsilon$ defined as
\begin{equation}
\VV_\epsilon(z)= \FF \left[ \epsilon \WW\left(\frac{z}{\epsilon}\right)+ \epsilon^2 \QQ\left(\epsilon;\frac{z}{\epsilon}\right) \right]-\FF(0)
\label{eq:defValt}
\end{equation}

\begin{lemma} \label{lemma:form}
In the framework above, for every wiggly potential $\VV_\epsilon$ of the form \eqref{eq:defValt} there exist two suitable functions $W$ and $Q$, such that $\VV_\epsilon$ can be written in the form \eqref{eq:defV} for sufficiently small $\epsilon>0$. Moreover we have $W(y)=\alpha \WW(y)$ and therefore
\begin{equation} \label{eq:fric_decomp}
\begin{array}{c}  \rho_+=\alpha \mu_+ \\ 
\rho_-=\alpha \mu_- 
\end{array}\quad \text{if $\alpha>0$}  \qquad 
\left(\text{resp.}  \quad
\begin{array}{c}  \rho_+=-\alpha \mu_- \\ 
\rho_-=-\alpha \mu_+ 
\end{array}\quad \text{if $\alpha<0$}  \right)
\end{equation}
\end{lemma}

\begin{proof}
We recall that, expanding $\FF$ as a Taylor series, we have
\begin{equation}
\FF(u)-\FF(0)=\alpha u + \frac{\FF''(0)}{2}u^2+h(u)u^2
\end{equation}
with $\displaystyle\lim_{u\to 0} h(u)=0$. Moreover, since $\FF\in \CC^3$, it can be shown that $h\in \CC^1$ and $h'(0)=\FF'''(0)/6$. Thus, applying this expansion to \eqref{eq:defValt}, we get
\begin{equation*}
\VV_\epsilon(z)=\epsilon W\left(\frac{z}{\epsilon}\right)+ \epsilon^2 Q\left(\epsilon;\frac{z}{\epsilon}\right)
\end{equation*}
where we set $W(y)=\alpha \WW(y)$ and
\begin{multline*}
Q(\epsilon ; y)=\alpha \QQ(\epsilon;y)+\frac{\FF''(0)}{2}\WW(y)^2+\epsilon^2\frac{\FF''(0)}{2} \QQ(\epsilon; y)^2+\\
+h\left( \epsilon \WW(y)+ \epsilon^2 \QQ(y) \right)\left[\WW(y)+ \epsilon \QQ(\epsilon;y) \right]^2 
\end{multline*}
All the desired properties of $W$ follow from their analogous ones for $\WW$.  To recover the desired estimates on $Q$ and $Q'$, we notice that, for any arbitrary $C_h>0$, we can find $\epsilon_h$ such that
\begin{equation*}
\begin{array}{l}
\abs{h\left( \epsilon \WW(y)+ \epsilon^2 \QQ(y) \right)}<C_h \\[0.2em]
\abs{h'\left( \epsilon \WW(y)+ \epsilon^2 \QQ(y) \right)}<\abs{\FF'''(0)}+1
\end{array}
  \qquad\text{for every $y\in \R$ and $\epsilon\in (0,\epsilon_h)$}
\end{equation*}
Thus, for every for every positive $\epsilon<\min\{1, \epsilon_\FF,\epsilon_\QQ, \epsilon_h\}$, we have
\begin{equation*}
\abs{Q(\epsilon;y)}\leq C_{Q,0}:=
\alpha \tilde C_{\QQ,0}+\frac{\FF''(0)}{2}\norm{\WW}_\infty^2+\frac{\FF''(0)}{2} \tilde C_{\QQ,0}^2
+C_h\left(\norm{\WW}_\infty+ \tilde C_{\QQ,0} \right)^2
\end{equation*}
The twice continuous differentiability of $Q$ follows from those of $\QQ$ and $\WW$, recalling also that $h(u)u^2$ is twice continuously differentiable  in $u$. Moreover we have the estimate
\begin{multline*}
\abs{Q'(\epsilon;y)}\leq C_{Q,1}:=
\alpha \tilde C_{\QQ,1}+\FF''(0)\norm{\WW}_\infty\norm{\WW'}_\infty+\FF''(0) \tilde C_{\QQ,0}\tilde C_{\QQ,1}+\\
+(\FF'''(0)+1)\left(\norm{\WW}_\infty+ \tilde C_{\QQ,0} \right)^2 +\\
+C_h\left(\norm{\WW}_\infty+ \tilde C_{\QQ,0} \right) \left(\norm{\WW'}_\infty+ \tilde C_{\QQ,1} \right) 
\end{multline*}

\end{proof}

The form \eqref{eq:defValt} of $\VV_\epsilon$ is interesting from a physical point of view, since it highlights the role of two different elements in our applications.
Formula \eqref{eq:fric_decomp} shows that the effective friction $\rho_\pm$ in the $\epsilon\to 0$ limit is the product of two quantities: $\mu_\pm$ associated to $\WW$ and $\alpha$ associated to $\FF$. 
 On one hand the ``geometric'' coefficients $\mu_+, \mu_-$ are related to the (directional) roughness of the surface, as perceived by the geometry of the system. On the other hand, the ``energetic''  coefficient $\alpha$ is associated to a \lq\lq tension\rq\rq\ in the element that mediates the frictional interaction.

This duality is quite central in our applications. Firstly, this distinction reinforces the resemblance with Coulomb's classical formulation of dry friction, where the friction intensity depends both on a coefficient, related to the properties of the interacting surfaces, and on the normal force exerted by each surface on the other one. Remarkably, in our models, the term $\alpha$ is proportional to the normal force exerted, in the limit case, on the surface by the mediating element.

Moreover, when discussing the \emph{with the nap/ against the nap} asymmetry in Section \ref{sec:contropelo}, we will see that it can be produced by two distinct effects: a geometric effect, given by the intrinsic asymmetry of the system, as captured by the coefficients $\mu_\pm$, and a energetic effect, where we observe a change of the configuration of the system between the two phases (with and against the nap), producing a change in the value of $\alpha$.

%%%%%%%%%%%%%%%%%%%%%%%%%%%%%%%%%%%%%%%%%%%%%%%%%%%%%%%%%%%%%%%%%%%%%%%%%%%%%%%%%%%%5

%%%%%%%		\MODELLING VARIO

%%%%%%%%%%%%%%%%%%%%%%%%%%%%%%%%%%%%%%%%%%%%%%%%%%%%%%%%%%%%%%%%%%%%%%%%%%%%%%%%%%%%%
\section{Modelling}

In this section we discuss three different models to obtain directional dry friction as the limit of the effects of an interaction with a surface having vanishingly small roughness, with the mediation of a hair/bristle-like element. We remind that, as in the previous section, we are assuming quasi-static evolution.

\paragraph{The limit system}
We characterize a frictional interaction governed by dry friction through a system, illustrated in Figure \ref{fig:limit}, consisting of a horizontal spring, that evolves as follows. The position of one end of the spring is controlled by the function $q\in \CC^1([0,T],\R)$; the second end of the spring, with position $u(t)$, is free to move and interacts with the surface, according to the force-velocity law
\begin{equation}
f^\mathrm{lim}(\dert u) =\begin{cases}
-\rho_+<0  &\text{if $\dert u> 0$}\\
\rho\in[-\rho_+,-\rho_-] &\text{if $\dert u= 0$}\\
-\rho_->0  &\text{if $\dert u< 0$}
\end{cases}
\end{equation} 
Thus the limit system has dissipation potential \eqref{eq:defRlim} and internal energy
\begin{equation}
\EE=\frac{k_h}{2}\left(L_h^\mathrm{rest} -q(t) +u \right)^2+\text{const.}
\end{equation}
where $k_h$ and $L_h^\mathrm{rest}$ are respectively  the elastic constant and the rest length of the spring. 

The state of the system will be described by a coordinate $z$ of the form $z(t)=u(t)+c$. The constant $c$, introduced for technical reasons, has different values in the models and can be thought as a gap between the position $u(t)$ of the second end of the spring and the position $z(t)$ at which, in the $\epsilon\to 0$ limit, the bristle-like mediating element interacts with the surface, cf.\ Figure \ref{fig:limit}. Thus the energy $\EE$ can be written in the form \eqref{eq:defE} by setting
\begin{align}
\Phi(z)=\frac{k_h}{2} z^2  && \ell(t)=k_h(q(t)-L_h^\mathrm{rest})
\end{align} 
and neglecting a remaining term $r(t)$, depending only on the time $t$, since it does not affect the dynamics \eqref{eq:din_lim}. We also remark that the change of variable to $z$ does not alter the dissipative terms, since $\dert u=\dert z$.

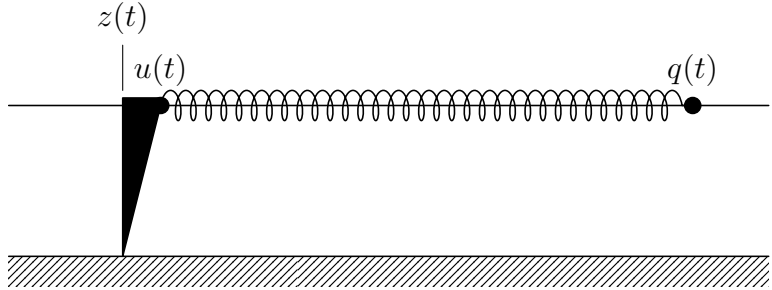
\begin{figure}\centering
\begin{tikzpicture}[line cap=round,line join=round,>=triangle 45,x=1.0cm,y=1.0cm, line width=0.6pt]
\clip(-1.,1.) rectangle (11.,6.);
\draw (10.,2.)-- (0.,2.);
\draw (0.,4.)-- (10.,4.);
\draw[thin] (1.5,4.2) --(1.5,4.8) node[anchor=south] {$z(t)$};
\draw[decoration={aspect=0.4, segment length=2mm, amplitude=2mm,coil},decorate] (2.,4.)-- (9.,4.) ;
\draw [fill=black] (2.,4.)--(2.,4.1)--(1.5,4.1)--   (1.5,4.)--(1.5,2)--(2.,4.);
\draw [fill=black] (2.,4.) circle (3pt) node[above=0.11cm]{$u(t)$};
\draw [fill=black] (9.,4.) circle (3pt) node[above=0.11cm]{$q(t)$};
\fill [pattern = north east lines] (0.,2) rectangle (10.,1.6);
\end{tikzpicture}
\caption{The limit system.}
\label{fig:limit}
\end{figure}

\paragraph{The approximating systems}

In the approximating systems, we imagine that the surface in no longer flat, but has a small, $\epsilon$-periodic perturbation of the form
\begin{equation}
w_\epsilon(x)=\epsilon w\left(\frac{x}{\epsilon}\right)
\end{equation}
where $w\in \CC^2(\R,\R)$ is a $1$-periodic (non-constant) function. Moreover we define
\begin{align}
\omega_+&=\max w'(x)>0  &
\omega_-&=\min  w'(x)<0 
\end{align}

The approximating systems are still characterized by a horizontal spring as in the limit model. However, the interaction with the surface is no longer subject to dry friction, but mediated by a new element, that ideally plays the role of a hair or a bristle, attached to the end $u(t)$ of the horizontal spring.
This element has, up to a constant, an internal energy $V_\epsilon$ as in \eqref{eq:defValt}, that depends only on $u$ and on the magnitude of the perturbation $\epsilon$. Finally, the only dissipative force acting on the system is a (vanishing) viscous force
\begin{equation}
f_\epsilon^\mathrm{vis}(\dert u)=-\epsilon^\gamma \dert u
\end{equation}
so that the Rayleigh dissipation potential of the system is given by \eqref{eq:defReps}.

\medbreak

In the following we discuss three different models for this mediating element.
In the first model, the mediating element is a vertical spring. The second model is actually a generalization of the first one, since in this case the spring forms a constant angle $\theta$ with the vertical axis. In the third model the mediating element is a straight rigid bar with constant length, but now the angle with the vertical axis can change and is influenced by an angular spring.

%%%%%%%%%%%%%%%%%%%%%%%%%%%%%%%%%%%%%%%%%%%%%%%%%%%%%%%%%%%%%%%%%%%%%%%%%%%%%%%%%%%%5

\subsection{First model: vertical spring}

\begin{figure}
\centering
\begin{tikzpicture}[line cap=round,line join=round,>=triangle 45,x=1.0cm,y=1.0cm, line width=0.6pt]
\clip(-1.,-1.) rectangle (11.,5.);
\draw (10.,0.)-- (0.,0.);
\draw (0.,4.)-- (10.,4.);
\draw[decoration={aspect=0.4, segment length=2mm, amplitude=2mm,coil},decorate] (4.,4.)-- (8.,4.) ;
\draw[<->] (0.5,4.)-- (0.5,0.);
\draw[line width=1.2pt,color=red,smooth,samples=100,domain=-0.5:10.5] plot(\x,{0.5*sin(((\x))*180/pi)}) node[anchor=south]{$w_\epsilon$} ;
\draw[decoration={aspect=0.5, segment length=2mm, amplitude=1.5mm,coil},decorate] (4.,4.)-- (4.,-0.378);
\draw [fill=black] (4.,4.) circle (3pt) node[above=0.11cm]{$u(t)$};
\draw [fill=black] (8.,4.) circle (3pt) node[above=0.11cm]{$q(t)$};
\draw [fill=black] (4.,-0.378) circle (3pt) ;
\draw(0.5,2) node[anchor=east] {$h$};
\end{tikzpicture}
\caption{First model: vertical spring}
\label{fig:model1}
\end{figure}
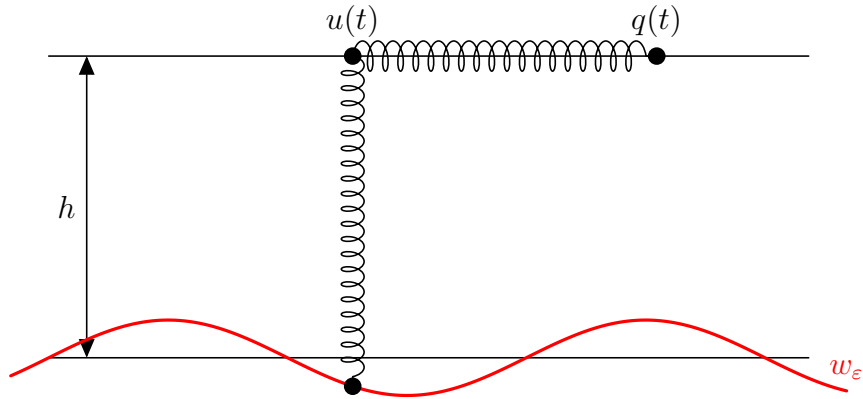

In our first model the mediating element (bristle) is a vertical spring, with horizontal position $u(t)$, as illustrated in Figure \ref{fig:model1}. One end of the spring has fixed height, while the height of the other end follows the fluctuation of the surface, in such a way that the length of the spring is 
\begin{equation*}
L(u)=h-w_\epsilon(u)
\end{equation*}
 Let $k>0$ and $L^\mathrm{rest}$ be respectively  the elastic constant and the rest length of the spring. Setting $z(t)=u(t)$, the energy of the vertical spring is
 \begin{equation*}
\frac{k}{2}\left(L^\mathrm{rest} -h+w_\epsilon(z)  \right)^2=\FF(\epsilon\WW(\frac{z}{\epsilon}))=V_\epsilon(z)+\FF(0)
\end{equation*}
where $\WW(y)=w(y)$ and $\FF(y)=\frac{k}{2}\left(L^\mathrm{rest} -h+y  \right)^2$, so that we have
\begin{align*}
\alpha&=\FF'(0)=k(L^\mathrm{rest}-h)  
& \mu_+&=\omega_+   & \mu_-&=\omega_-
\end{align*}
We require
\begin{equation*}
L^\mathrm{rest}\neq h
\end{equation*}
so that $\alpha\neq 0$ and \eqref{eq:cond_alpha} is satisfied. We notice that, for instance, setting $L^\mathrm{rest}> h$ means that the spring is always compressed.

In this way all the requirements of Lemma \ref{lemma:form} are satisfied, and therefore we can apply Theorem \ref{th:main} to obtain the desired behaviour for the limit system. In this way, for a compressed spring, we recover a sort of Coulomb law, since the friction coefficients are proportional to the normal force exerted by the spring on the surface, that, in the limit, is exactly equal to $\alpha$. Moreover, if the profile of the fluctuations is asymmetric, in the sense that $\omega_+\neq \omega_-$, then also the friction is asymmetric.

%%%%%%%		SECONDO MODELLO

%%%%%%%%%%%%%%%%%%%%%%%%%%%%%%%%%%%%%%%%%%%%%%%%%%%%%%%%%%%%%%%%%%%%%%%%%%%%%%%%%%%%%

\subsection{Second model: slanted spring}

\begin{figure}
\centering
\begin{tikzpicture}[line cap=round,line join=round,>=triangle 45,x=1.0cm,y=1.0cm, line width=0.6pt]
\clip(-1.,-1.) rectangle (11.,5.);
\draw (10.,0.)-- (0.,0.);
\draw (0.,4.)-- (10.,4.);
\draw (4.,4.)-- (4.,2.5);
\draw[decoration={aspect=0.4, segment length=2mm, amplitude=2mm,coil},decorate] (4.,4.)-- (8.,4.) ;
\draw[<->] (0.5,4.)-- (0.5,0.);
\draw[line width=1.2pt,color=red,smooth,samples=100,domain=-0.5:10.5] plot(\x,{0.5*sin(((\x))*180/pi)}) node[anchor=south]{$w_\epsilon$} ;
\draw[decoration={aspect=0.5, segment length=2mm, amplitude=1.5mm,coil},decorate] (4.,4.)-- (2.,0.45);
\draw (2.,0.45)-- (2.,0.) node[anchor=north]{$p(t)$};
\draw [shift={(4.,4.)},fill=blue,fill opacity=0.2]  (0,0) --  plot[domain=4.198:4.712,variable=\t]({1.*1.117*cos(\t r)+0.*1.117*sin(\t r)},{0.*1.117*cos(\t r)+1.*1.117*sin(\t r)}) -- cycle ;
\draw[color=blue] (3.7,2.5) node {$\theta$};
\draw [fill=black] (4.,4.) circle (3pt) node[above=0.11cm]{$u(t)$};
\draw [fill=black] (8.,4.) circle (3pt) node[above=0.11cm]{$q(t)$};
\draw [fill=black] (2.,0.45) circle (3pt) ;
\draw(0.5,2) node[anchor=east] {$h$};
\end{tikzpicture}
\caption{Second model: slanted spring}
\label{fig:model2}
\end{figure}
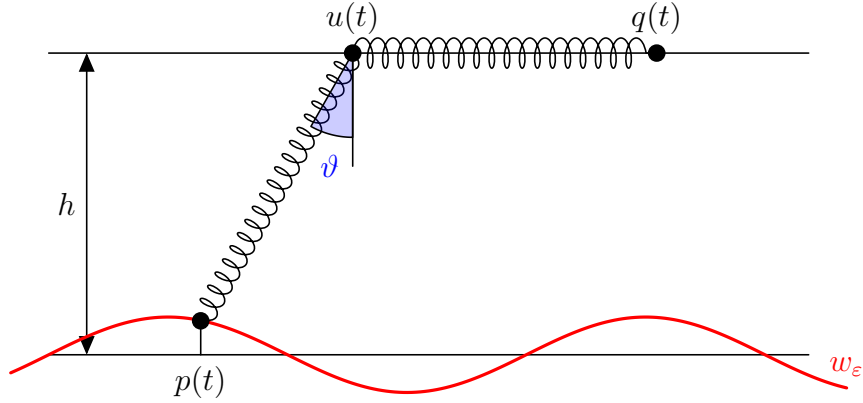

Our second model generalizes the first one, since in this case we consider a slanted spring forming a fixed angle $0<\theta<\pi/2$ with the vertical axis, as illustrated in Figure \ref{fig:model2}. As before, one end of the spring has fixed height and horizontal position $u(t)$. In this case, however, the horizontal position of the second end will be different from $u(t)$ and denoted with $p(t)$. We therefore have
\begin{equation}
\frac{u-p}{h-w_\epsilon(p)}=\tan \theta
\end{equation}
We can express explicitly $u$ as a function of  $p$ as
\begin{equation} \label{eq:IIz(p)}
u - h\tan \theta=p-w_\epsilon(p)\tan \theta 
\end{equation}
We require 
\begin{equation} \label{eq:m2_rough}
\omega_+<\cot \theta
\end{equation}
so that  $w_\epsilon '(p)\tan \theta<\omega_+\tan \theta<1$ and therefore $p_\epsilon(u)$ is a one-to-one correspondence. The length of the spring is thus
\begin{equation}
L=\sqrt{(u-p)^2+(h-w_\epsilon(p))^2}=\frac{u-p}{\sin \theta}=\frac{h-w_\epsilon(p)}{\cos \theta}
\end{equation}
For our purposes, it is convenient to adopt the variable $$z=u-h\tan \theta$$ to represent the state of the system. Setting
\begin{equation*}
g(p)=p-w(p)\tan \theta
\end{equation*}
we notice, for every choice of $\epsilon>0$,  the function $g$ relates $z(t)$ with $p(t)$ through the one-to-one correspondences 
\begin{equation*}
\frac{z(t)}{\epsilon}= g\left(\frac{p(t)}{\epsilon}\right)
\end{equation*}
The bijectivity of $g$ follows from \eqref{eq:m2_rough} since
\begin{equation*}
g'(p)=1-w'(p)\tan \theta>0
\end{equation*}
The inverse function $g^{-1}$ is twice continuously differentiable and such that $g^{-1}(z+1)=g^{-1}(z)+1$ for every $z\in\R$.

\medbreak

We set
\begin{align*}
\WW(z)=w\left(g^{-1}(z)\right)  && \QQ(\epsilon;z)=0
\end{align*}
and
\begin{equation*}
\FF(y)=\frac{k}{2}\left(L^\mathrm{rest} -\frac{h-y}{\cos \theta} \right)^2
\end{equation*}
so that, up to a constant, the internal energy of the slanted spring is given by $\VV_{\epsilon}(z)=\FF(\epsilon\WW(z/\epsilon))$ of the form \eqref{eq:defValt}.

\medbreak

We now want to determine the coefficients $\rho_\pm$. Since
\begin{equation} \label{eq:m2_alpha}
\alpha= \frac{k}{\cos \theta}\left( L^\mathrm{rest}-\frac{h}{\cos \theta}\right)
\end{equation}
it remains to find $\mu_+$ and $\mu_-$. Since this involves the derivative of $g^{-1}$, difficulties may arise trying a direct computation,   since $g$ cannot  be always inverted explicitly and thus, in general, $\WW$ may not be explicitly determined.
Such is the case, for instance, of a sinusoidal choice of $w$, for which the inversion of $g$ leads to the well studied problem of the inverse Kepler equation~\cite{Arn07}. 

However, for our purpose, the full knowledge of the fluctuation profile  as perceived by the slanted spring, i.e.~the explicit form of $\WW$, is not necessary, since we are only interested in the minimum and maximum of $\WW'$. Such values can be computed without inverting $g$ explicitly.
Since the same issue will arise also in the next model, we summarize the result in the following lemma.

\begin{lemma} \label{lem:mu}
Let $w\in\CC^2(\R,\R)$, be a $1$-periodic function with $\omega_+=\max w'(z)>0$  and $\omega_-=\min  w'(z)<0$ . For some constant $a$, with $\omega_-^{-1}<-a<\omega_+^{-1}$, we consider
\begin{align*}
g(p)=p+aw(p)>0   &&  \WW(z)=w\left(g^{-1}(z)\right)
\end{align*}
Then
\begin{align}
\mu_+&=\max \WW'(z)=\frac{\omega_+}{1+a\omega_+}  &
\mu_-&=\min  \WW'(z)=\frac{\omega_-}{1+a\omega_-}
\end{align}
\end{lemma}

\begin{proof}
For any fixed $\bar z\in \R$, let us define $\bar p=g^{-1}(\bar z)$. We have
\begin{equation}
\WW'(\bar z)=w '(\bar p)\cdot (g^{-1})'(\bar z)=w '(\bar p)\frac{1}{g'(\bar p)}=\frac{w'(\bar p)}{1+ a w'(\bar p)}
\end{equation}
Since $g$ is a bijection and the function $\displaystyle y\mapsto \frac{y}{1+ay}$ is increasing monotone for $y\in[\omega_-,\omega_+]$, we get
\begin{align*}
\mu_+&=\max_{\bar z\in \R} W'(\bar z)=\max_{\bar p\in\R} \frac{w'( \bar p)}{1+ a w'(\bar p)}=\frac{\omega_+}{1+a\omega_+} \\[0.2em]
\mu_-&=\min_{\bar z\in \R}  W'(\bar z)=\min_{\bar p\in\R} \frac{w'( \bar p)}{1+ a w'(\bar p)}=\frac{\omega_-}{1+a\omega_-}
\end{align*}
\end{proof}

Thus, for our second model, we have
\begin{align}
\mu_+&=\frac{\omega_+}{1-\omega_+\tan \theta}  &
\mu_-&=\frac{\omega_-}{1-\omega_-\tan \theta }
\end{align}
We notice that,  for $\theta=0$, we recover the situation of the first model, as expected. The behaviour of the coefficient as function of $\theta$ is illustrated in Figure \ref{fig:model2_mu}.

\begin{figure}
\centering
\includegraphics[]{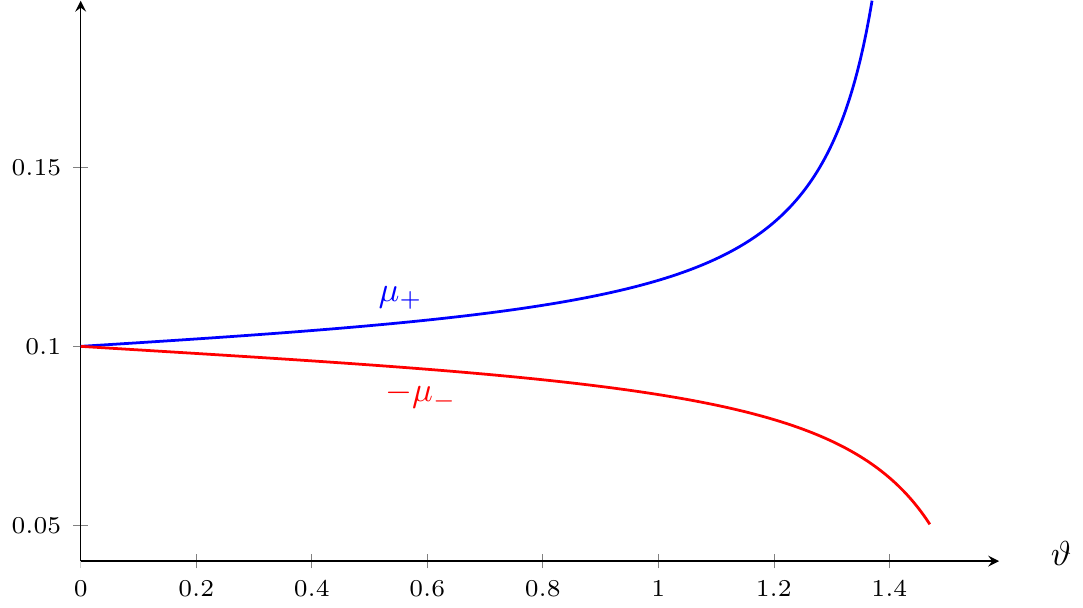}
    \caption{Behaviour of $\mu_+$ and $\mu_-$ in the second model as a function of $\theta$. We are setting $\omega_+=-\omega_-=0.1$, so that by \eqref{eq:m2_rough} the admissible domain is $0< \theta< \arccot  0.1$.}
\label{fig:model2_mu}
\end{figure}

\medbreak

Thus all the requirements of Lemma \ref{lemma:form} are satisfied, and  Theorem \ref{th:main} can be applied. We also observe from \eqref{eq:m2_alpha} that the coefficient $\alpha$ is proportional to the normal force exerted by the spring on flat surface at $\epsilon=0$, by a factor $1/\cos^2\theta$.

\medbreak

We notice that, for this model, we have $\rho_+>-\rho_-$, meaning that the friction opposing a rightward movement ($\dert u>0$) is greater than the one corresponding to a leftward movement ($\dert u<0$). This is exactly the opposite of what we usually experience in the \emph{with the nap/against the nap} asymmetry, for which, as we will discuss in Section \ref{sec:contropelo}, other explanations can be found.

Remarkably, such a \lq\lq reversed\rq\rq\ with the nap/against the nap asymmetry has been observed in experiments dealing with friction force microscopy on molecular monolayers \cite{KwaShi04,Lil98}; the resemblance with such situations suggests a possible connection.

%%%%%%%%%%%%%%%%%%%%%%%%%%%%%%%%%%%%%%%%%%%%%%%%%%%%%%%%%%%%%%%%%%%%%%%%%%%%%%%%%%%%5

%%%%%%%		TERZO MODELLO

%%%%%%%%%%%%%%%%%%%%%%%%%%%%%%%%%%%%%%%%%%%%%%%%%%%%%%%%%%%%%%%%%%%%%%%%%%%%%%%%%%%%%

\subsection{Third model: angular spring} \label{sec:3m}

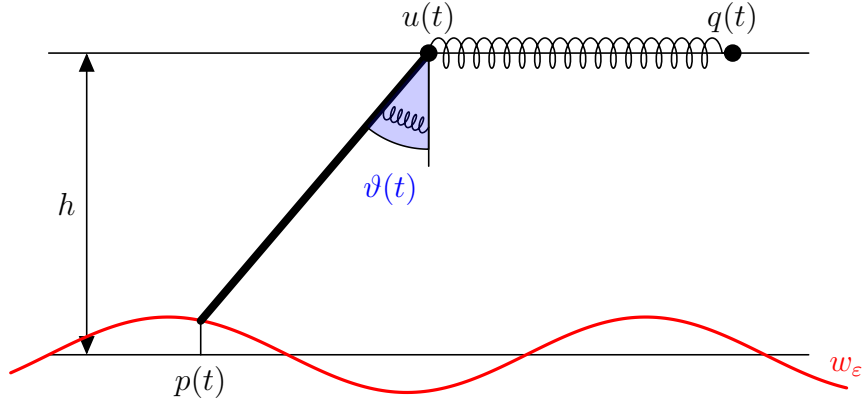
\begin{figure}
\centering
\begin{tikzpicture}[line cap=round,line join=round,>=triangle 45,x=1.0cm,y=1.0cm, line width=0.6pt]
\clip(-1.,-1.) rectangle (11.,5.);
\draw (10.,0.)-- (0.,0.);
\draw (0.,4.)-- (10.,4.);
\draw (5.,4.)-- (5.,2.5);
\draw[decoration={aspect=0.4, segment length=2mm, amplitude=2mm,coil},decorate] (5.,4.)-- (9.,4.) ;
\draw[decoration={aspect=0.4, segment length=1mm, amplitude=1mm,coil},decorate] (5.,3.016797393581479)--(4.364897125418177,3.2494457354953767);
\draw[<->] (0.5,4.)-- (0.5,0.);
\draw[line width=1.2pt,color=red,smooth,samples=100,domain=-0.5:10.5] plot(\x,{0.5*sin(((\x))*180/pi)}) node[anchor=south]{$w_\epsilon$} ;
\draw[line width= 3pt] (5.,4.)-- (2.,0.45);
\draw (2.,0.45)-- (2.,0.) node[anchor=north]{$p(t)$};
\draw [shift={(5.,4.)},fill=blue,fill opacity=0.2]  (0,0) --  plot[domain=4.01:4.712,variable=\t]({1.*1.274*cos(\t r)+0.*1.2744*sin(\t r)},{0.*1.274*cos(\t r)+1.*1.2744373933593*sin(\t r)}) -- cycle ;
\draw[color=blue] (4.5,2.2) node {$\theta(t)$};
\draw [fill=black] (5.,4.) circle (3pt) node[above=0.11cm]{$u(t)$};
\draw [fill=black] (9.,4.) circle (3pt) node[above=0.11cm]{$q(t)$};
\draw(0.5,2) node[anchor=east] {$h$};
\end{tikzpicture}
\caption{Third model: angular spring, for $\theta^\mathrm{rest}=0$}
\label{fig:model3}
\end{figure}

In this model, the mediating element consists of a straight rigid rod with length $L$, as illustrated in Figure \ref{fig:model3}.
One end of the rod has constant height and horizontal position $u(t)$. The rod can rotate around this end and we denote with $\theta>0$ the angle formed with the vertical axis. We denote with $p(t)$ the horizontal coordinate of the second end of the rod and assume that the systems is oriented so that $p<u$.
Denoting with $h$ the distance between  the first end of the rod and the limit flat surface, we require $L>h$, so that, for sufficiently small oscillations $w_\epsilon$, the rod can always touch the surface. We define
\begin{equation}
\theta^\mathrm{lim}=\arccos \frac{h}{L}>0
\end{equation}
as the angle of the rod when it touches the flat surfaces in the limit $\epsilon\to 0$. 
The rod has an angular spring with rest angle $\theta^\mathrm{rest}$. We assume
\begin{equation} \label{eq:m3_angle_ass}
\theta^\mathrm{lim} > \theta^\mathrm{rest} >-\frac{\pi}{2}
\end{equation}

The internal energy of the spring is
\begin{equation*}
\frac{k}{2}\left(\theta - \theta^\mathrm{rest} \right)^2
\end{equation*}
Since the surface acts as a constraint on the system and we consider quasi-static motion, for each value of $u$, we deduce that the rod assumes the minimum angle possible $\theta=\theta (u)$, touching the surface.

We require, for every $x\in \R$,
\begin{equation} \label{eq:m3_cont_ass}
-\tan \theta^\mathrm{lim} < w'(x) < \cot \theta^\mathrm{lim}
\end{equation}
For $\epsilon$ sufficiently small, the right inequality assures that the rod touches the wiggly surface only with its second end, whereas the left inequality implies that, when $u(t)$ changes, then also $p(t)$ changes, but without jumps.
Thus, since the second end of the rod touches the surface, we can deduce the following relationships:
\begin{equation} \label{eq:m3_coord}
L=\sqrt{(u-p)^2+(h-w_\epsilon(p))^2}=\frac{u-p}{\sin \theta}=\frac{h-w_\epsilon(p)}{\cos \theta}
\end{equation}
From this, can express explicitly $\theta(t)$ as a function of $p(t)$, namely
\begin{equation}
\theta(t)=\arccos \frac{h-w_\epsilon(p(t))}{L}
\end{equation}

\medbreak

Let us introduce the new variable
\begin{equation}
z(t)=u(t)-\sqrt{L^2-h^2}
\end{equation} 
We now want to show that $w_\epsilon(p(t))$ can be expressed as a function of $z(t)$ of the form
\begin{equation*}
w_\epsilon(p(t))=\epsilon \WW\left(\frac{z(t)}{\epsilon}\right)+ \epsilon^2 \QQ\left(\epsilon;\frac{z(t)}{\epsilon}\right)
\end{equation*}
with $\WW$ and $\QQ$ as in \eqref{eq:defValt}; in this way also $\theta(t)$ can be expressed  as a function of $z(t)$.

\medbreak

From \eqref{eq:m3_coord} we can express $z(t)$ as a function of $p(t)$, as
\begin{equation} \label{eq:m3_zp}
z(t)=p(t)+A(w_\epsilon(p(t)))
\end{equation}
where
$$
A(y)=\sqrt{L^2-(h-y)^2}-\sqrt{L^2-h^2}
$$
We notice that $A(0)=0$ and $A'(y)=\frac{h-y}{\sqrt{L^2-(h-y)^2}}$. Equation \eqref{eq:m3_zp} gives a one-to-one correspondence between $z(t)$ and $p(t)$, since
\begin{equation} \label{eq:m3_der}
\Der{z}{p}=1+A'(w_\epsilon(p))w'\left(\frac{p}{\epsilon}\right)=1+(\cot \theta) w'\left(\frac{p}{\epsilon}\right)>0
\end{equation}
for $\epsilon$ sufficiently small. 
The last inequality follows from the fact that, for $\epsilon\to 0$, we have $\norm{w_\epsilon}_\infty\to 0$ and $\theta \approx \theta^{lim}$. Hence,  by \eqref{eq:m3_cont_ass}, we can find $\epsilon_\theta$ such that, for $\epsilon<\epsilon_\theta$, we always have $\cot \theta>0$.

\medbreak

Let us denote $Z=z/\epsilon$ and $P=p/\epsilon$. From \eqref{eq:m3_zp}, we get a twice continuously differentiable bijection $Z=G(\epsilon;P)$,  that can be decomposed as
$$
G(\epsilon;P)=G_0(P)+\epsilon G_R(\epsilon;P)
$$
where
$$
G_0(P)=P+\frac{h}{\sqrt{L^2-h^2}}w(P)=P+(\cot \theta^\mathrm{lim})w(P)
$$
and $G_R(\epsilon;P)$ is $1$-periodic and twice continuously differentiable in $P$; moreover $G_R$ and its derivative in $P$ are uniformly bounded for $\epsilon$ sufficiently small.

From \eqref{eq:m3_der} we know that $D_P \, G(\epsilon; P)>0$ for every $P\in \R$; thus, for each $\epsilon<\epsilon_\theta$, the function $G(\epsilon;\cdot)$ has a twice continuously differentiable  inverse $H(\epsilon; \cdot)$, so that $P=H(\epsilon,Z)$. The function $H$ can be written in the form 
$$
H(\epsilon;Z)=H_0(Z)+\epsilon H_R(\epsilon;Z)
$$
Here $H_0$ is twice continuously differentiable, $1$-periodic in $Z$ and there are two positive constants $C_H$ and $\epsilon_H$ such that
 \begin{equation*}
\begin{array}{l}
\abs{H_R\left( \epsilon; Z \right)}<C_H \\[0.2em]
\abs{D_Z H_R\left( \epsilon; Z \right)}<C_H
\end{array}
  \qquad\text{for every $Z\in \R$ and every $\epsilon\in (0,\epsilon_H)$}
\end{equation*} 
A straightforward computation shows that $H_0=G_0^{-1}$.

\medbreak

Let us notice that, since that, since $w$ is periodic and twice continuously differentiable, there exists a continuously differentiable function $h_w\colon \R\times \R \to R$,  $1$-periodic and such that
$$
w(x+\epsilon)=w(x)+\epsilon h_w(\epsilon; x)
$$
Moreover there exist two positive constants $C_w$ and $\epsilon_w$ such that 
 \begin{equation*}
\begin{array}{l}
\abs{h_w\left( \epsilon; x \right)}<C_w \\[0.2em]
\abs{D_x h_w\left( \epsilon; x \right)}<C_w
\end{array}
  \qquad\text{for every $x\in \R$ and every $\epsilon\in (0,\epsilon_w)$}
\end{equation*} 
Thus we have
\begin{align*}
w (P)&= w\left(H_0(Z)+\epsilon H_R(\epsilon;Z)\right)\\
&=w(H_0(Z))+\epsilon h_w(\epsilon H_R(\epsilon;Z);H_0(Z))H_R(\epsilon;Z)
\end{align*}

\medbreak

We set
\begin{align*}
&\WW(y)=w(H_0(y))\\[0.1em]
&\QQ(\epsilon; y)=h_w(\epsilon H_R(\epsilon;y);H_0(y))H_R(\epsilon;y)\\[0.1em]
&\FF(y)=\frac{k}{2}\left(\arccos \frac{h-y}{L} - \theta^\mathrm{rest} \right)^2
\end{align*}
and observe the energy of the angular spring is, up to a constant, expressed by a function $\VV_\epsilon(z)$ of the form \eqref{eq:defValt}, with constants $\tilde C_{\QQ,0}=C_wC_H$, $\tilde C_{\QQ,1}=C_wC_H(\norm{H_0'}_\infty
+1)$ and $\epsilon_\QQ=\min\{\epsilon_\theta, \epsilon_H,\epsilon_w\}$.

We obtain that
\begin{equation} \label{eq:m3_alpha}
\alpha=\FF'(0)=\frac{k}{\sqrt{L^2-h^2}}(\theta^\mathrm{lim} - \theta^\mathrm{rest} )
\end{equation}
so that, by \eqref{eq:m3_angle_ass}, the assumption \eqref{eq:cond_alpha} is satisfied, as are also the other requirements of Lemma \ref{lemma:form}. Thus Theorem \ref{th:main} gives the desired behaviour for $\epsilon\to 0$.

As in the previous model, in general $G$ cannot be inverted explicitly. However we can apply Lemma \ref{lem:mu} to recover the coefficients $\mu_+,\mu_-$. We have
\begin{align} \label{eq:m3_mu}
\mu_+&=\frac{\omega_+}{1+\omega_+\cot \theta^\mathrm{lim}}  &
\mu_-&=\frac{\omega_-}{1+ \omega_-\cot \theta^\mathrm{lim}}
\end{align}
where we recall that $\cot \theta^\mathrm{lim}=\frac{h}{\sqrt{L^2-h^2}}$. The behaviour of the coefficient as a function of $\theta^\mathrm{lim}$ is illustrated in Figure \ref{fig:model3_mu}.

\begin{figure}
\centering

\includegraphics[]{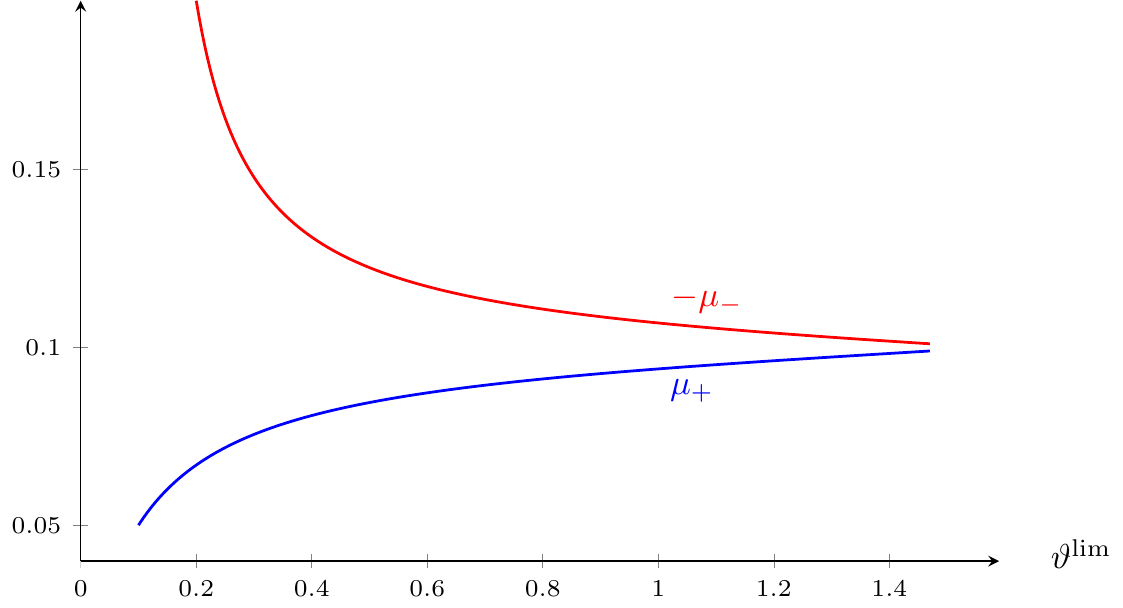}

\caption{Behaviour of $\mu_+$ and $\mu_-$ in the third model as a function of $\theta^\mathrm{lim}$. We are setting $\omega_+=-\omega_-=0.1$, so that by \eqref{eq:m3_cont_ass} the admissible domain is $\arctan 0.1 < \theta^\mathrm{lim} < \arccot  0.1$.}
\label{fig:model3_mu}
\end{figure}

%%%%%%%%%%%%%%%%%%%%%%%%%%%%%%%%%%%%%%%%%%%%%%%%%%%%%%%%%%%%%%%%%%%5

%%%%%%%%%%%%%%%%%%%%%%%%%%%%%%%%%%%%%%%%%%%%%%%%%%%%%%%%%5

\subsection{Interpretation of the \emph{with the nap/against the nap} effect.} \label{sec:contropelo}

A hairy surface is a common denominator of many situations where we  experience a directionality in the friction:  stroking a cat, rubbing a brush with slanted bristles, using climbing skins for backcountry skiing or brushing napped fabric. Although we intuitively gather all this instances under the same name of \emph{with the nap/against the nap} asymmetry,
what we are actually considering is family of different phenomena, all producing the same kind of directional effect. For instance, in some situations there is no significant change in the bristle configuration between the two phases (e.g.,~rubbing gently a hard brush), while in others large deformations of the bristles occur and we observe a dramatic change in their configuration passing from one direction to the other one (e.g.,~stroking a cat).

 Clearly a comprehensive and complete characterization  of all these with the nap/against the nap phenomena would require a sophisticated modelling of the mechanical behaviour of a bristle. Yet, with the help of the model of Section \ref{sec:3m}, we can easily identify two fundamental effects that are involved.
The geometric one is a direct application of the angular spring model, and holds for sufficiently rigid bristles, remaining straight also under small compressions.
The energetic one instead applies to flexible bristles, buckling very easily when compressed.

\paragraph{Geometric  effect}
From \eqref{eq:m3_mu}, we obtain that, for the angular spring model, we have $\rho_+>-\rho_-$, meaning that the friction opposing a rightward movement ($\dert u>0$) is smaller than the one corresponding to a leftward movement ($\dert u<0$). This is exactly what we expect by the with the nap/against the nap effect.

However, it is not obvious that a bristle should always behave as a rigid bar with an angular spring. Indeed, especially during strokes against the nap, the rod is subject to a longitudinal compression, that could produce buckling in a flexible bar, invalidating the model. An estimate of the axial tension along the bar, obtained by considering the limit case when the lower end of the bar moves on a flat surface experiencing dry friction,  is
\begin{equation}\label{eq:nap_tens}
T=-\frac{k}{L}(\theta^\mathrm{lim}- \theta^\mathrm{rest}) \cot \theta^\mathrm{lim}
+\frac{\rho_\pm}{\sin \theta^\mathrm{lim}}
\end{equation}
where $\rho_\pm$ depends on the direction of motion. We observe that during a stroke against the nap (so $\rho_\pm=\rho_-<0$) the bar is always compressed ($T<0$), however this tension is small when the bristle oscillates near its rest position ($\theta^\mathrm{lim}\approx \theta^\mathrm{rest}$) and the friction coefficients are small. This situation suits well to the motion of a hard brush rubbed gently on a smooth surface.

\paragraph{Energetic  effect} 
When the critical load for buckling is too low, the above description is no longer valid, but we can still apply the model of Section \ref{sec:3m} when the bristle is subject to traction. The following  interpretation of the with the nap/against the nap effect  is based on such assumption.

When moved with the nap, the hair is rotated in the same direction of its rest angle, as shown in Figure \ref{fig:contropelo}(a), so that the angular spring is only slightly stretched. On the other hand, when the hair is moved against the nap, it is rotated in the opposite direction of its rest position, as shown in Figure \ref{fig:contropelo}(b); in this way the angular spring is much more stretched than in the previous case. Another way to describe this scenario is to notice that the tip of the hair is always behind its root, with respect to the direction of motion.

\begin{figure}[t]
\centering
\subfloat[\emph{With the nap}.]{
\begin{tikzpicture}[line cap=round,line join=round,>=triangle 45,x=1.0cm,y=1.0cm, scale=0.6]
\clip(-0.5,-2) rectangle (9.5,7);
\draw (0.,0.)-- (9.,0.);
\draw (0.,4.)-- (9.,4.);
\draw[very thick, <-] (8.,6.)-- (5.,6.) node[anchor=east]{Motion};
\draw[decoration={aspect=0.4, segment length=1.5mm, amplitude=1.5mm,coil},decorate] (5.,4.)-- (8.,4.) ;
\draw [fill=black] (5.,4.) circle (3pt) node[above=0.07cm]{$u(t)$};
\draw [fill=black] (8.,4.) circle (3pt) node[above=0.07cm]{$q(t)$};
\draw[line width= 3pt] (5.,4.)-- (1.,0.);
\draw [shift={(5.,4.)},fill=blue,fill opacity=0.2] (0,0) --  plot[domain=4.346:4.71,variable=\t]({1.*1.5*cos(\t r)+0.*1.5*sin(\t r)},{0.*1.5*cos(\t r)+1.*1.5*sin(\t r)}) -- cycle ;
\draw[dashed, very thick] (5.,4.)-- (2.976,-1.282);
\draw (5.,4.)-- (5.,2.3);
\draw[decoration={aspect=0.4, segment length=0.6mm, amplitude=0.6mm,coil},decorate] (3.94,2.94)-- (4.525,2.577);
\fill [pattern = crosshatch] (0.,0) rectangle (9.,-0.3);
\draw[color=blue] (4.3,1.8) node[anchor=west] {$\theta_\mathrm{with}>0$};
\end{tikzpicture}} \qquad
\subfloat[\emph{Against the nap}.]{
\begin{tikzpicture}[line cap=round,line join=round,>=triangle 45,x=1.0cm,y=1.0cm, scale=0.6]
\clip(-0.5,-2) rectangle (9.5,7);
\draw (0.,0.)-- (9.,0.);
\draw (0.,4.)-- (9.,4.);
\draw[very thick, <-] (8.,6.)-- (5.,6.) node[anchor=east]{Motion};
\draw[decoration={aspect=0.4, segment length=1.5mm, amplitude=1.5mm,coil},decorate] (5.,4.)-- (8.,4.) ;
\draw [fill=black] (5.,4.) circle (3pt) node[above=0.07cm]{$u(t)$};
\draw [fill=black] (8.,4.) circle (3pt) node[above=0.07cm]{$q(t)$};
\draw[line width= 3pt] (5.,4.)-- (1.,0.);
\draw [shift={(5.,4.)},fill=blue,fill opacity=0.2]  (0,0) --  plot[domain=4.712:5.08,variable=\t]({1.*1.5*cos(\t r)+0.*1.5*sin(\t r)},{0.*1.5*cos(\t r)+1.*1.5*sin(\t r)}) -- cycle ;
\draw (5.,4.)-- (5.,2.3);
\draw[dashed, very thick] (5.,4.)-- (7.02,-1.28);
\draw[decoration={aspect=0.4, segment length=1mm, amplitude=0.6mm,coil},decorate] (3.94,2.94)-- (5.474,2.57);
\fill [pattern = crosshatch] (0.,0) rectangle (9.,-0.3);
\draw[color=blue] (4.7,1.8) node[anchor=west] {$\theta_\mathrm{against}<0$};
\end{tikzpicture}}
\caption{Energetic interpretation of the \emph{with the nap/against the nap} asymmetry. The dashed line represent the rest angle of the bar.}
\label{fig:contropelo}
\end{figure}
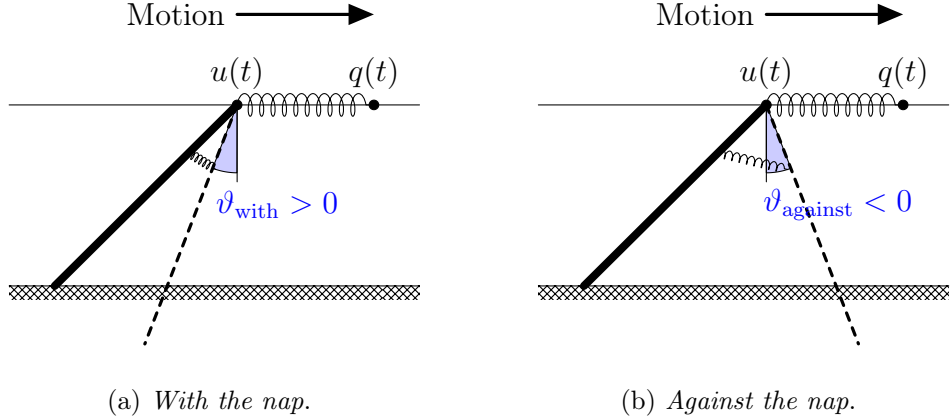

Hence, if we report both situations to the framework of our model of Section \ref{sec:3m} (as done in Figure \ref{fig:contropelo}), we observe that the two cases share the same coefficient $\mu_+$, while we have a change in the coefficient $\alpha$, since the rest angle of the hair changes. In case of \emph{with the nap} motion, the rest angle of the hair is $\theta_\mathrm{with}>0$. On the other hand, the case of \emph{against the nap} motion corresponds to $\theta_\mathrm{against}=-\theta_\mathrm{with}<0$. 

In this way we can immediately recover the friction coefficients using \eqref{eq:m3_alpha} and \eqref{eq:m3_mu}. We get
\begin{align*}
\rho_\mathrm{with}&=\frac{\mu_+}{\tan \theta^\mathrm{lim}} (\theta^\mathrm{lim}-\theta_\mathrm{with})\\
\rho_\mathrm{against}&=\frac{\mu_+}{\tan \theta^\mathrm{lim}} (\theta^\mathrm{lim}-\theta_\mathrm{against})=\frac{\mu_+}{\tan \theta^\mathrm{lim}} (\theta^\mathrm{lim}+\theta_\mathrm{with})
\end{align*}
where we trivially have $\rho_\mathrm{against}>\rho_\mathrm{with}$,
in agreement with our common experience of the phenomenon.

\medbreak

We now analyse the compatibility of this interpretation with the tension of the bristle during the motion. We notice that, since in both phases we have $\rho_\pm=\rho_+$, the last term in \eqref{eq:nap_tens} gives always a positive contribution to the tension. Thus, if the surface is quite rough ($\rho_+$ large) and the bristle flexible ($k/L$ small), the rod is subject to traction, so that our model provides a good approximation. 

We remark that this energetic interpretation requires a transitional phase, where the bristle is strongly deformed, to account for the change of configuration occurring when the direction of motion is inverted.  Since we assume a small critical load and high friction, we expect this transition to be triggered by buckling when the direction is changed, and that, afterwards, a sufficiently long motion in same same direction restores the bristle to a stable straight state, as those we discussed above.

\section{Convergence structure} \label{sec:strat} 
The main issue in Section \ref{sec:abstr} is the change in the nature of the dissipation: in the approximating systems $(\EE_\epsilon,\Rc_\epsilon)$ we have a viscous drag (i.e.~the dissipation potential $\Rc_\epsilon$ is quadratic), whereas in the limit system it is rate independent (i.e.~the dissipation potential $\Rc_\epsilon$ is positively homogeneous of degree $1$).  
Such a situation has been successfully addressed in continuum mechanics, showing that rate-independent plasticity can be obtained as limit of a chain of viscous bistable springs \cite{MieTru12,PugTru05}.
Here we follow the recent approach by Mielke \cite{Mie12} (cf.~also \cite{Mie15}), based on the De Giorgi's $(\Rc,\Rc^*)$ formulation, also called \emph{energy-dissipation principle}. 

\medbreak

We begin by recalling some known facts about the Legendre transform (cf.~for instance \cite{RocWet}). 

\paragraph{Legendre transform and De Giorgi's $(\Rc,\Rc^*)$ formulation} 
Let us consider a function $\Psi\colon \R\to \R\cup \{+\infty\}$ that is proper (i.e.~not identically $+\infty$), lower semi-continuous and convex .  The Legendre transform $\Psi^*\colon \R \to \R\cup \{+\infty\}$ of $\Psi$ is defined as
\begin{equation*}
\Psi^*(\xi)=\sup_{x\in X} \left[{\xi}{x}-\Psi(x)\right]
\end{equation*}
The function $\Psi^*$  is proper, lower semi-continuous and convex; moreover we have  $(\Psi^*)^*=\Psi$.

We now briefly recall some well-known properties of the Legendre transform. 
The Fenchel estimate states that, for every $x\in \R$ and $\xi \in X\R$,  we have
\begin{equation}
\Psi(x)+\Psi^*(\xi)\geq {\xi}{x}
\label{eq:FenEst}
\end{equation}
The case when the equality holds is characterized by the Legendre-Fenchel equivalence:
\begin{align}
\xi\in \partial \Psi(x) 
&& \Longleftrightarrow  &&
x\in \partial \Psi^*(\xi) 
&&\Longleftrightarrow  &&
\Psi(x)+\Psi^*(\xi)= {\xi}{x}
\label{eq:LegFenEq}
\end{align}

Let us now consider the problem
\begin{equation}
0\in \partial_{\dert z} \tilde \Rc({\dert z})+D_z\tilde \EE(t,z)
\label{eq:din_tilde}
\end{equation}
where $\tilde \EE\in \CC^1([0,T]\times \R,\R)$ and $\tilde{\Rc}\colon \R\to \R$ is a convex function. A solution of the problem is a function $z\colon [0,T]\to \R$ that satisfies \eqref{eq:din_tilde} for almost every $t\in [0,T]$. Note that this framework covers both the wiggly systems \eqref{eq:din_eps} and the limit system \eqref{eq:din_lim}.

Let us therefore define $\tilde \Rc^*\colon \R \to \R\cup \{+\infty\}$ as the Legendre transform of the function $\tilde\Rc$.  
First of all, let us notice that, by the Legendre-Fenchel equivalence \eqref{eq:LegFenEq}, the inclusion \eqref{eq:din_tilde} is equivalent to
%\begin{equation}
%\dert z(t)\in \partial \tilde \Rc^*\left(-D_z\tilde \EE(t,z(t)) \right)
%\label{eq:din_tilde_tras2}
%\end{equation}
\begin{equation}
\left(\dert z(t), -D_z\tilde \EE(t,z(t)) \right) \in \tilde \CC_{\Psi+\Psi^*} =\left\{(x,\xi) \colon \Psi(x)+\Psi^*(\xi)= {\xi}{x} \right\}
\label{eq:din_tilde_tras}
\end{equation}

De Giorgi's  $(\Rc,\Rc^*)$ formulation of the problem consists in the following sufficient condition for being a solution of \eqref{eq:din_tilde}.

\begin{prop} \label{prop:DGform}
A function $z\colon [0,T]\to \R$ is a solution of \eqref{eq:din_tilde} if and only if it satisfies
\begin{equation}
\begin{split}
\tilde \EE(T,z(T))+\int_0^T{\left[\tilde \Rc(\dert z(s))+\tilde \Rc^*\left(-D_z\tilde \EE(s,z(s))\right)\right]\dd s} \leq\\
\leq \tilde \EE (0,z(0)) +\int_0^T{\partial_t\tilde\EE(s,z(s))}\dd s
\end{split}
\label{eq:DGform}
\end{equation}
\end{prop}

We now prove a slightly more general proposition, suitable to our purposes.
Let us replace the integral dissipation term in the left-hand side of \eqref{eq:DGform} with a term of the form
\begin{equation} \label{eq:dissipForm}
\tilde \DD(z)=\int_0^T \tilde \MM(\dert z(s), -D_z\tilde \EE(s,z))\dd s
\end{equation}
where we require
\begin{equation}
\tilde \MM(x,\xi)\geq {\xi}{x} \qquad \text{for every $x\in\R$, $\xi\in\R$}
\label{eq:DGfen_gen}
\end{equation}
Moreover let us define the set
\begin{equation}
\tilde \CC_\MM =\left\{(x,\xi) \colon \tilde \MM(x,\xi)= {\xi}{x} \right\}
\end{equation}

\begin{prop} \label{prop:DGform_gen}
A function $z\colon [0,T]\to \R$ satisfies
\begin{equation}
\tilde \EE(T,z(T))+\tilde \DD(z)
\leq \tilde \EE (0,z(0)) +\int_0^T{\partial_t\tilde\EE(s,z(s))}\dd s
\label{eq:DGform_gen}
\end{equation}
if and only if it satisfies 
\begin{equation}
\left(\dert z(t), -D_z\tilde \EE(t,z(t)) \right) \in \tilde \CC_\MM \qquad \text{for almost every $t\in[0,T]$}
\end{equation}
\end{prop}

\begin{proof}
Using the chain rule, we get that the estimate \eqref{eq:DGform_gen} is equivalent to
\begin{equation*}
\int_0^T{ \tilde \MM(\dert z(s), -D_z\tilde \EE(s,z))\dd s}\leq 
-\int_0^T{{D_z\tilde \EE(s,z(s))}{\dert z(s)}\dd s}
\end{equation*} 
Looking at  estimate \eqref{eq:DGfen_gen}, we get that \eqref{eq:DGform_gen} is true if and only if  equality in \eqref{eq:DGfen_gen} holds for almost every $t\in[0,T]$.
\end{proof}

We remark that Proposition \ref{prop:DGform_gen} applies to the case $\tilde \MM(x,\xi)=\Rc(x)+\Rc^*(\xi)$.  Thus Proposition \ref{prop:DGform} follows as an immediate corollary,   since, as we have seen, the Legendre-Fenchel equivalence implies the equivalence between \eqref{eq:din_tilde} and \eqref{eq:din_tilde_tras}. 

%%%%%%%%%%%%%%%%%%%%%%%%%

\paragraph{Convergence structure}

Our strategy to prove Theorem \ref{th:main} is  to consider the convergence of the systems only once they have been reformulated in the form \eqref{eq:DGform}.

Let us consider a family of energy functions $\EE_\epsilon\in\CC^1([0,T]\times\R,\R)$, and the corresponding dissipation functionals  $\DD_\epsilon$ of the form
\begin{equation} 
 \DD_\epsilon(z)=\int_0^T \MM_\epsilon(\dert z(s), -D_z \EE_\epsilon(s,z))\dd s
\end{equation}
where $\MM_\epsilon(x,\xi)\geq {\xi}{x}$ for every $x\in\R$, $\xi\in\R$.
We are given a family of functions $z_\epsilon\colon [0,T]\to \R$ that solve the associated evolution problems, i.e.~each $z_\epsilon$ satisfies the estimate
\begin{equation} \label{eq:DGest_eps}
 \EE_\epsilon(T,z(T))+\DD_\epsilon(z)
\leq  \EE_\epsilon (0,z(0)) +\int_0^T{\partial_t\EE_\epsilon(s,z(s))}\dd s
\end{equation}
Then, we consider a limit  energy function $\EE\in\CC^1([0,T]\times\R,\R)$, and a limit dissipation functional $\DD$ of the form
\begin{equation} 
 \DD(z)=\int_0^T \MM(\dert z(s), -D_z \EE(s,z))\dd s
\end{equation}
where $\MM(x,\xi)\geq {\xi}{x}$ for every $x\in\R$, $\xi\in\R$. We define
\begin{equation}
\CC_\MM =\left\{(x,\xi) \colon \MM(x,\xi)= {\xi}{x} \right\}
\end{equation}

\begin{prop} \label{prop:conv}
Let $ \EE_\epsilon, \DD_\epsilon, z_\epsilon, \EE$ and $\DD$ be as above. Assume that there exists a continuous function $\bar z\colon [0,T]\to R$ such that $z_\epsilon\to \bar z$ in $\CC([0,T],\R)$. Suppose that, for every $t\in[0,T]$, the following estimates hold
\begin{gather*}
\EE(t,\bar z(t))\leq \liminf_{\epsilon\to 0} \EE_\epsilon(t,z_\epsilon(t))\\
\partial_t\EE(t,\bar z(t))=\lim_{\epsilon\to 0} \partial_t\EE_\epsilon(t,z_\epsilon(t)) \\
\DD(\bar z)\leq \liminf_{\epsilon\to 0} \DD_\epsilon(z_\epsilon)
\end{gather*}
and moreover
\begin{equation}
\EE(0,\bar z(0))=\lim_{\epsilon\to 0} \EE_\epsilon(0,z_\epsilon(0)) 
\end{equation}
Then $\bar z$ is a solution of the problem
\begin{equation} \label{eq:din_Clim} 
\left(\dert{ \bar z}(t), -D_z\tilde \EE(t,\bar z(t)) \right) \in  \CC_\MM \qquad \text{for almost every $t\in[0,T]$}
\end{equation}
\end{prop}
\begin{proof}
The convergence assumptions, applied to \eqref{eq:DGest_eps}, lead to the estimate
\begin{equation}
\EE(T,\bar z(T))+ \DD(\bar z)\leq \EE (0,\bar z(0)) +\int_0^T{\partial_t\EE(s,\bar z(s))}\dd s
\end{equation}
The thesis follows from Proposition \ref{prop:DGform_gen}.
\end{proof}

%%%%%%%%%%%%%%%%%%%%%%%%%%%%%%%%%%%%%%%%%%%%%%%%%%%%%%%%%%%%%%%%%%%%%%%%%%%%%%%%%%%%5

%%%%%%%		PROOF OF THE MAIN THEOREM

%%%%%%%%%%%%%%%%%%%%%%%%%%%%%%%%%%%%%%%%%%%%%%%%%%%%%%%%%%%%%%%%%%%%%%%%%%%%%%%%%%%%%

\section{Proof of Theorem \ref{th:main} } \label{sec:proof}

We now implement the convergence strategy of Section \ref{sec:strat} to the situation described in Section \ref{sec:abstr}, in order to prove Theorem \ref{th:main}. From now on the symbols $\EE, \EE_\epsilon, \Rc, \Rc_\epsilon$, etc.~have the same properties and meaning considered in Section \ref{sec:abstr}. In addition, we assume that the hypothesis \eqref{eq:cond_zetazero} of Theorem \ref{th:main} holds. 

\medbreak

Our plan is to apply Proposition \ref{prop:conv}. To begin, we recall the definition \eqref{eq:defReps} of $\Rc_\epsilon$ and define $\DD_\epsilon$ and $\MM_\epsilon$  by setting
\begin{equation}
\MM_\epsilon(v,\xi)=\Rc_\epsilon(v)+\Rc_\epsilon^*(\xi)
=\frac{\epsilon^\gamma v^2}{2}+\frac{\xi^2}{2\epsilon^\gamma}
\end{equation}
By Proposition \ref{prop:DGform}, each function $z_\epsilon$ satisfies the estimate \eqref{eq:DGest_eps}.

\medbreak

The reformulation of the limit system requires a little more attention. Let us first define the set $\Omega_0=[\rho_-,\rho_+]$ and denote, for any set $A\subseteq \R$,
\begin{equation}
\mychi_A(\xi)=\begin{cases}
0 &\text{for $\xi\in A$}\\
+\infty &\text{for $\xi\notin A$}
\end{cases}
\end{equation}
To define the functions $\DD$ and $\MM$, instead of the trivial choice associated to De Giorgi's formulation of problem \eqref{eq:din_lim}, we set
\begin{equation} \label{eq:defMlim}
\MM(v,\xi)=\abs{v}K(\xi)+\mychi_{\Omega_0}(\xi)
\end{equation}
where \begin{equation}
K(\xi)=\int_0^1{\abs{\xi-W'(y)}\dd y}
\end{equation}
Since $W'$ is continuous, $1$-periodic with zero average and has image $\Omega_0$, we deduce that $K(\xi)>\abs{\xi}$ if $\xi \in \inter \Omega_0$, whereas $K(\xi)=\abs{\xi}$ if  $\xi \notin \inter \Omega_0$. As a consequence, we obtain the desired estimate $\MM(x,\xi)\geq {\xi}{x}$. Moreover we have
\begin{equation}
\CC_\MM=\left( \{0\} \times \Omega_0 \right) \cup \left( (-\infty,0) \times \{\rho_-\} \right)\cup \left( (0,+\infty) \times \{\rho_+\} \right)
\end{equation}

Recalling the definition \eqref{eq:defRlim} of $\Rc$, we have $\Rc^*(\xi)=\mychi_{\Omega_0}(\xi)$, and so $\CC_\MM=\CC_{\Rc+\Rc^*}$. This means that, by Proposition \ref{prop:DGform}, problem \eqref{eq:din_Clim} is equivalent to \eqref{eq:din_lim}.

\medbreak

To apply Proposition \ref{prop:conv} and complete the proof of Theorem \ref{th:main}, it is left to prove
\begin{itemize}
\item  the existence of a limit function $\bar z$, such that $z_\epsilon\to \bar z$ in $\CC([0,T],\R)$;
\item that the estimate $\displaystyle\DD(\bar z)\leq \liminf_{\epsilon\to 0} \DD_\epsilon(z_\epsilon)$ holds.
\end{itemize} 
These will be the subjects of the next two subsections.

\subsection{Convergence of the solutions} \label{subsec:solut}

\paragraph{Preliminary notation}

Without loss of generality, we restrict our discussion to the interval $\epsilon\in (0,\bar \epsilon]$, where $\bar \epsilon$ is sufficiently small to satisfy $\bar \epsilon<\min \{1, \epsilon_Q\}$. We set $\beta=\min\{1,\gamma\}$ and notice that, for the values of $\epsilon$ considered, we have $\epsilon^\beta=\max\{\epsilon,\epsilon^{\gamma}\}$. 

Let us also introduce the following notations for some recurrent constants.
We call $\Lambda_\ell$  the Lipschitz constant of $\ell$.  By the uniform convexity of $\Phi$, we can find a constant $\phi>0$ such that $\Phi''(z)>\phi$ for all $z\in\R$. Since  $W$ and its first derivative are bounded, we denote
\begin{align*}
C_{W,0}&=\norm{W}_\infty &
C_{W,1}&=\norm{W'}_\infty 
\end{align*}

%%%%%%%%%%%%

\paragraph{Strip of admissible solutions}
Let us now define $\tilde z_\pm \colon [0,T]\to \R$ as 
\begin{align}
\tilde z_-(t)&=(\Phi')^{-1}(\ell (t)-\rho_+) & \tilde z_+(t)&=(\Phi')^{-1}(\ell (t)-\rho_-)
\end{align}
We recall that this definition is well-posed since, by the uniform convexity of $\Phi$, $\Phi'$ is globally invertible and $\Imm \Phi'=\R$.  Since the image of $\ell$ is bounded, by compactness arguments, we also have
\begin{equation*}
C_{\pm}=\max\{\norm{\dert{\tilde z}_+}_\infty, \norm{\dert{\tilde z}_-}_\infty\} <+\infty
\end{equation*}
We notice that condition \eqref{eq:cond_zetazero} can be restated by writing $z^0\in [\tilde z_-(0),\tilde z_+(0)]$. Moreover, looking carefully at the inclusion \eqref{eq:din_lim}, we observe that the solution $\bar{z}$ is bounded between $\tilde z_-$ and $\tilde z_+$, and the current state can possibly  change (i.e. $\dert{\bar z}(t)\neq 0$) only if $\bar{z}=\tilde z_-$ (and therefore $\dert{\bar z}(t)\geq 0$) or $\bar{z}=\tilde z_+$ (and therefore $\dert{\bar z}(t)\leq 0$). The strip $[\tilde z_-(t),\tilde z_+(t)]$ gives the evolution of the elastic domains of the limit system.

Hence, we define the distance at each time $t$ of a solution $z_\epsilon$ of \eqref{eq:din_eps} from this region, by setting
\begin{equation*}
\delta_\epsilon(t)=\dist\left(z_\epsilon(t),[\tilde z_-(t),\tilde z_+(t)]\right)
\end{equation*}
Notice that \eqref{eq:cond_zetazero} implies $\delta_\epsilon(0)\to 0$ for $\epsilon\to 0$.

%%%%%%%%%%%%%%%%%%%%%%%%%%%%%

\paragraph{Estimates on $z_\epsilon$} Let us recall that, by \eqref{eq:din_eps}, the solution $z_\epsilon$ satisfies
\begin{equation}
\epsilon^\gamma \dert z_\epsilon(t)= -\Phi'(z_\epsilon(t)) -W'\left(\frac{z_\epsilon(t)}{\epsilon}\right) - \epsilon Q'\left(\epsilon;\frac{z_\epsilon(t)}{\epsilon}\right) +\ell(t)
\end{equation}
The value of $\delta_\epsilon(t)$ is controlled by the following estimate.
\begin{lemma} \label{lemma:est1} There exists a constant $C_0>0$ such that, for every $t\in[0,T]$ and $\epsilon\in(0,\bar \epsilon)$, we have
\begin{equation}
\delta_\epsilon(t)\leq \delta_\epsilon(0)e^{{-\phi t}/{\epsilon^\gamma}}+\epsilon^\beta C_0  
\end{equation}
Moreover, if $t\in [0,T]$ is such that $\delta_\epsilon(t)>\epsilon^\beta C_0$, we have that $\dert \delta_\epsilon (t)<0$.
\end{lemma}

\begin{proof}
If $z_\epsilon(t)\in(\tilde z_-(t),\tilde z_+(t))$, then the estimate follows immediately. 
Let us now consider the case $z_\epsilon(t)\geq \tilde z_+(t)$.
We have
\begin{align*}
\epsilon^\gamma \dert \delta_\epsilon 
&= \epsilon^\gamma \dert z_\epsilon - \epsilon^\gamma \dert{\tilde z}_+ \\
&\leq -\Phi'(z_\epsilon) -W'\left(\frac{z}{\epsilon}\right) - \epsilon Q'\left(\epsilon;\frac{z}{\epsilon}\right)  +\Phi'(\tilde z_+)+\rho_- +\epsilon^\gamma C_\pm\\
&\leq -\Phi'(z_\epsilon)+\Phi'(\tilde z_+) +\epsilon C_{Q,1} +\epsilon^\gamma C_\pm\\
&\leq -\phi\delta_\epsilon +\epsilon C_{Q,1} +\epsilon^\gamma C_\pm\\
&\leq -\phi\delta_\epsilon + C_1\epsilon^\beta
\end{align*}
where $C_1=C_\pm + C_{Q,1}$. The same estimate can be obtained analogously in the case $z_\epsilon(t)\geq \tilde z_+(t)$. Thus the required estimate for $\delta_\epsilon$ follows, by a suitable application of Gronwall's Lemma, for with $C_0=C_1/\phi$.
\end{proof}

Let us notice that, combining Lemma \ref{lemma:est1} with assumption \eqref{eq:cond_zetazero} and the Lipschitz continuity of $\tilde{z}_\pm$, it can be shown that all the solutions $z_\epsilon$ are bounded within an interval $[z_\mathrm{min},z_\mathrm{max}]$. By compactness, in this interval the function $\Phi'$ is Lipschitz continuous with Lipschitz constant $\Lambda_{\Phi'}$.

\begin{lemma} \label{lemma:est2}
For every $C_2>0$, there exists $C_3>0$ such that, for every $\epsilon\in (0,\bar \epsilon)$  and every solutions $z_\epsilon$ of  \eqref{eq:din_eps}, if  
\begin{equation} \label{eq:est2ass}
\delta_\epsilon(t_0 )\leq \epsilon^\beta C_2 \qquad \text{for some $t_0\in [0,T]$}
\end{equation}
then
\begin{equation*}
\abs{z_\epsilon(t)-z_\epsilon (t_0)}\leq \epsilon^\beta  C_3 \qquad \text{for every $t\in I_\epsilon^0=[t_0,t_0+\epsilon^\beta]\cap [0,T]$}
\end{equation*}
\end{lemma}

\begin{proof}
Let us set
\begin{equation*}
b_\epsilon(z)=-\Phi'(z)-V'_\epsilon(z)
\end{equation*}
We plan to find two points $\zeta_-$ and $\zeta_+$ such that
\begin{equation*}
z_\epsilon (t_0)-\epsilon^\beta C_3  \leq \zeta_- \leq z_\epsilon (t_0)
\leq \zeta_+ \leq z_\epsilon (t_0)+\epsilon^\beta C_3 
\end{equation*}
and, for every $t\in I_\epsilon^0$,
\begin{equation*}
b_\epsilon(\zeta_-)+\ell(t)>0  \qquad \text{and} \qquad
b_\epsilon(\zeta_+)+\ell(t)< 0
\end{equation*}
This last condition implies that every solution of \eqref{eq:din_eps} starting at $t_0$ inside the interval $[\zeta_-,\zeta_+]$ cannot cross its boundary in the time interval $I^0_\epsilon$. 

\medskip

We present the proof only for $\zeta_+$, since $\zeta_-$ can be found similarly. Let $y_+\in \R$ be any point such that $W'(y_+)=\rho_+$; we will look for $\zeta_+\in y_+ + \epsilon\Z$, so that $W'(\zeta_+)=\rho_+$.
We know that, for every $t\in I_\epsilon^0$,
\begin{align*}
\ell(t) &\leq \ell(t_0)+\Lambda_\ell (t-t_0)\leq \Phi'(\tilde z_-(t_0))+\rho_+ +\epsilon^\beta\Lambda_\ell  \\
&\leq \Phi'(z_\epsilon(t_0))+\epsilon^\beta\Lambda_{\Phi'} C_2  +\rho_+ +\epsilon^\beta\Lambda_\ell 
\end{align*}
Thus we have
\begin{align*}
b_\epsilon(\zeta_+)+\ell(t) &\leq -\Phi'(\zeta_+)-\rho_+ +\epsilon C_{Q,1} 
 +\Phi'(z_\epsilon(t_0))+\epsilon^\beta C_2  \Lambda_{\Phi'} +\rho_+ +\epsilon^\beta \Lambda_\ell   \\
 &\leq -(\zeta_+-z_\epsilon(t_0))\phi +\epsilon^\beta C_4
\end{align*}
where $C_4=C_{Q,1} +\Lambda_{\Phi'} C_2+\Lambda_\ell$. Therefore we take the smallest value $\zeta_+\in y_+ + \epsilon\Z$  satisfying $\zeta_+> z_\epsilon(t_0) +\frac{\epsilon^\beta C_4}{\phi}$.  This choice gives one part of the thesis with $C_3=1+C_4/\phi$.

We proceed similarly for $\zeta_-$  and conclude the proof.
\end{proof}

\begin{lemma} \label{lemma:est3}
There exists a constant $C>0$ such that, for every $s,t\in[0,T]$, the following estimate holds:
\begin{gather}
\abs{z_\epsilon(t)-z_\epsilon(s)}\leq C(\delta_\epsilon(0)+\abs{t-s}+\epsilon^\beta) \label{eq:est3} 
\end{gather}
\end{lemma}

\begin{proof}
Using Lemma \ref{lemma:est1} we can characterize the possible behaviours of $z_\epsilon$. 

If $\delta_\epsilon(0)\leq 2\epsilon^\beta C_0$, then $\delta_\epsilon(t)\leq 2\epsilon^\beta C_0$ for every $t\in[0,T]$. In this case the assumptions \eqref{eq:est2ass} of Lemma \ref{lemma:est2} are satisfied for every $t_0\in [0,T]$ by taking $C_2=2C_0$. Now, for every $s,t\in [0,T]$, we set $k\in \N$ such that $\abs{t-s}/\epsilon^\beta\leq k< \abs{t-s}/\epsilon^\beta +1$. We can therefore construct a partition $s=\tau_0<\tau_1<\dots <\tau_{k-1}<\tau_k=t$, such that, for every $i=1, \dots, k$, we have $\tau_i-\tau_{i-1}<\epsilon^\beta$. Thus we have
\begin{equation}
\abs{z_\epsilon(t)-z_\epsilon(s)}\leq \sum_{i=1}^k \abs{z_\epsilon(\tau_i)-z_\epsilon(\tau_{i-1})}\leq C_3 k\epsilon^\beta \leq C_3 (\abs{t-s}+\epsilon^\beta)
\label{eq:est3part1}
\end{equation}
where $C_3$ is given by Lemma \ref{lemma:est2} and does not depend on $\epsilon$.

On the other hand, if $\delta_\epsilon(0)> 2\epsilon^\beta C_0$,  Lemma \ref{lemma:est1} shows that the solution $z_\epsilon$ monotonically gets closer to the strip $[\tilde z_-, \tilde z_+]$, and \emph{possibly} at some time $t_\epsilon$ satisfies $\delta_\epsilon(t_\epsilon)=2\epsilon^\beta C_0$, so that $\delta_\epsilon(t)\leq 2\epsilon^\beta C_0$ for every $t\in [t_\epsilon,T]$. 
For $s,t\in [0,t_\epsilon]$ (or in $[0,T]$ if there is no such $t_\epsilon$), we have the estimate
\begin{equation}
\abs{z_\epsilon(t)-z_\epsilon(s)}\leq \delta_\epsilon(0)+C_\pm \abs{t-s}
\label{eq:est3part2}
\end{equation}
If there is a $t_\epsilon\in [0,T]$ as above, since $\delta_\epsilon(t)\leq 2\epsilon^\beta C_0$ for every $t\in [t_\epsilon,T]$ we can proceed as in the first part of the proof and the estimate \eqref{eq:est3part1} holds for every $s,t\in[t_\epsilon,T]$. 

We set $C=C_3+C_\pm+1$ and the proof is completed by combining \eqref{eq:est3part1} and \eqref{eq:est3part2}, possibly splitting the estimate in two parts if $s<t_\epsilon<t$.
\end{proof}

%%%%%%%%%%%%%%%%%%%%%%%%%%%%%

\paragraph{Convergence of the solutions $z_\epsilon$}
By Lemma \ref{lemma:est3} we obtain the equicontinuity of the family of functions $z_\epsilon\colon [0,T]\to \R$, for $\epsilon\in (0,\bar \epsilon]$. 
By the Ascoli-Arzelà Theorem we can find a subsequence ${(z_{\epsilon_i})}_{i\in\N}$ with $\epsilon_i\to 0$ for which there exists a continuous function $\bar z \colon [0,T]\to \R$ such that $z_{\epsilon_i}\to \bar z$ uniformly in $\CC([0,T])$. 
A second consequence of Lemma \ref{lemma:est3} is that $\bar z$ is Lipschitz continuous with constant $C$, that is $\abs{\bar z(t)-\bar{z}(s)}<C\abs{t-s}$ for every $s,t\in[0,T]$.

It remains to show that $\bar z$ is a solution of \eqref{eq:din_lim} and that actually  the whole sequence $z_\epsilon$ converges to $\bar z$, not only a subsequence $ z_{\epsilon_i}$. We will address these issues in Subsection \ref{subsec_completeproof}.

%%%%%%%%%%%%%%%%%%%%%%%%%%%%%

\subsection{Estimate on the dissipation functionals}
Let us write
\begin{align*}
\eta_\epsilon(t) &=-\Phi'(z_\epsilon(t))+\ell(t)\\
u_\epsilon (t) &=W'\left(\frac{z_\epsilon(t)}{\epsilon}\right)+\epsilon Q'\left(\epsilon ;\frac{z_\epsilon(t)}{\epsilon}\right)\\
\xi_\epsilon (t) &=\eta_\epsilon (t)-u_\epsilon (t)
\end{align*}

\begin{lemma} \label{lemma:dis}
Let $z_\epsilon,\bar z\in W^{1,1}([0,T])$ and $\eta_\epsilon,\bar \eta\in \CC^0([0,T])$ be such that, for $\epsilon\to 0$,
\begin{equation*}
z_\epsilon\to \bar z \qquad\text{and} \qquad \eta_\epsilon\to \bar \eta \qquad\text{in $\CC^{0}([0,T])$}.
\end{equation*}
Then
\begin{equation}
\liminf_{\epsilon\to_0} \int_0^T{\MM_\epsilon\bigl(\dert z_\epsilon(t), \xi_\epsilon(t)\bigr)\dd t}
\geq \int_0^T \MM\bigl(\dert{\bar z}(t), \eta(t)\bigr)\dd t
\end{equation}
\end{lemma}

\begin{proof}
Let us define the interval
\begin{equation*}
\Omega_\epsilon=\left[\rho_- -\epsilon C_{Q,1}\, , \, \rho_+ +\epsilon C_{Q,1} \right]
\end{equation*}
so that $u_\epsilon\in \Omega_\epsilon$ and $\Omega_0=[\rho_-,\rho_+]$, as defined above. We recall that
$\xi_\epsilon =\eta_\epsilon-u_\epsilon$, implying $\abs{\xi_\epsilon}\geq \dist(\eta_\epsilon, \Omega_\epsilon)$.

\medbreak

We therefore obtain the following lower bound for $\MM_\epsilon$:
\begin{equation} \label{eq:estDiss1}
\begin{split}
\MM_\epsilon(v,\xi_\epsilon)
&=\frac{\epsilon^\gamma v^2}{2}+\frac{(1-\epsilon^{\frac{\gamma}{2}})\xi_\epsilon^2}{2\epsilon^\gamma}+\frac{\epsilon^{\frac{\gamma}{2}}\xi_\epsilon^2}{2\epsilon^\gamma}\\
&\geq (1-\epsilon^{\frac{\gamma}{2}})\abs{v}\abs{\xi_\epsilon}
+\frac{1}{2\epsilon^{\frac{\gamma}{2}}}\left[\dist (\eta_\epsilon, \Omega_\epsilon)\right]^2
\end{split}
\end{equation}
We now derive two separate estimates for the two terms of the right hand side of \eqref{eq:estDiss1}.

\medbreak

For the second term, we observe that
\begin{equation}
\liminf_{\epsilon\to 0} \frac{1}{2\epsilon^{\frac{\gamma}{2}}}\left[\dist (\eta_\epsilon, \Omega_\epsilon)\right]^2 \geq \mychi_{\Omega_0}(\eta)
\end{equation}
so, by Fatou's Lemma, we obtain
\begin{equation} \label{eq:estDissSPart2}
\liminf_{\epsilon\to 0} \int_0^T{\frac{1}{2\epsilon^{\frac{\gamma}{2}}}\left[\dist (\eta_\epsilon(t), \Omega_\epsilon)\right]^2 \dd t} 
\geq \int_0^T{\mychi_{\Omega_0}(\eta(t))\dd t}
\end{equation}

\medbreak

To study the integral of the remaining term in \eqref{eq:estDiss1}, let us consider the integral
\begin{equation}
\De =\int_0^T{\abs{\dert z_\epsilon(t)}\abs{\xi_\epsilon(t)}\dd t}
\end{equation}
We define, for every integer $n>(\bar\epsilon)^{-1}$ and $j\in\{1,2,\dots,n\}$, the time interval
\begin{equation}
I_j^n=\left[\frac{j-1}{n}T, \frac{j}{n}T \right)
\end{equation}
to which we associate the value
\begin{equation*}
h_j^n(y)=\inf\left\{\abs{\eta_{\tilde\epsilon}(s) -W'(y)-\tilde \epsilon Q'(\tilde\epsilon;y)},\quad \text{for $s\in I^n_j$, $\tilde \epsilon\in \left(0,\frac{1}{n}\right)$} \right\}
\end{equation*}
We remark that $h^n_j$ is periodic with period $1$.
We also notice that, by definition, for every $t\in I^n_j$ and every $\epsilon\in \left(0,\frac{1}{n}\right)$ we have $\abs{\xi_\epsilon(t)}>h^n_j\left(\frac{z_\epsilon(t)}{\epsilon}\right)$.
Thus, for each $\epsilon<\frac{1}{n}$,
\begin{equation*}
\De \geq \sum_{j=1}^n \,\int_{I^n_j}{ \abs{\dert z_\epsilon (t)} h^n_j\left(\frac{z_\epsilon(t)}{\epsilon}\right) \dd t}
\end{equation*}
Let us now consider the case $z(\frac{j-1}{n}T)<z(\frac{j}{n}T)$. We have
\begin{multline*}
\int_{I^n_j}{ \abs{\dert z_\epsilon (t)} h^n_j\left(\frac{z_\epsilon(t)}{\epsilon}\right) \dd t}
\geq \int_{z_\epsilon(\frac{j-1}{n}T)}^{z_\epsilon(\frac{j}{n}T)}h^n_j\left(\frac{z}{\epsilon}\right) \dd z\\
\quad\xrightarrow{\epsilon\to 0} \quad
\left[z\left(\tfrac{j}{n}T\right) - z\left(\tfrac{j-1}{n}T\right) \right]
\int_{0}^{1}h^n_j(y) \dd y
\end{multline*}
since, due to the periodicity of $h^n_j$, for $\epsilon\to 0$ the integral of $h^n_j(z/\epsilon)$ on a given interval tends to the integral of the average value of $h^n_j$. Arguing similarly for $z(\frac{j-1}{n}T)>z(\frac{j}{n}T)$, we get
\begin{equation}
\liminf_{\epsilon\to 0} \int_{I^n_j}{ \abs{\dert z_\epsilon (t)} h^n_j\left(\frac{z_\epsilon(t)}{\epsilon}\right) \dd t}
\geq \abs{z\left(\tfrac{j}{n}T\right) - z\left(\tfrac{j-1}{n}T\right) }
\int_{0}^{1}h^n_j(y) \dd y
\label{eq:estDissDim1}
\end{equation}

\medbreak

Let $z_n$ be the piecewise affine interpolant such that $z_n(\frac{j}{n}T)=z(\frac{j}{n}T)$ for every $j+0,1,\dots,n$.
We define $k_n(t)$ as the average of $h^n_j$, where $j$ is the one such that $t\in I^n_j$, that means
\begin{equation*}
k_n(t)=\int_{0}^{1}h^n_j(y) \dd y \qquad \text{for $\frac{j-1}{n}T\leq t< \frac{j}{n}T$}
\end{equation*}
Thus, summing the estimates \eqref{eq:estDissDim1} for $j+1,\dots,n$, we get
\begin{equation*}
\liminf_{\epsilon\to 0} \De 
\geq \sum_{j=1}^n 
 \abs{z\left(\tfrac{j}{n}T\right) - z\left(\tfrac{j-1}{n}T\right) }
\int_{0}^{1}h^n_j(y) \dd y
=\int_0^T k_n(t)\abs{\dert z_n(t)} \dd t
\end{equation*}

Since we assumed that $z\in W^{1,1}([0,T])$, we know that $\dert z_n \to \dert z$ strongly in $L^1([0,T])$ for $n\to \infty$. Moreover, the uniform convergence $(\eta_\epsilon,z_\epsilon)\to (\eta,\bar z)$ assures us that $k_n(t)\to K(\eta(t))$ uniformly.
Thus we get
\begin{equation} \label{eq:estDissSPart1}
\liminf_{\epsilon \to 0} \De \geq \int_0^T \abs{\dert z(t)} K(\eta(t)) \dd t
\end{equation}

The proof is completed combining the estimates \eqref{eq:estDissSPart2} and \eqref{eq:estDissSPart1}.
\end{proof}

%%%%%%%%%%%%%%%%%%%%%%%%%%%%%%%%%%%%%%%%%%%%%%%%%%%

\subsection{Completion of the proof} \label{subsec_completeproof}

At the end of Subsection \ref{subsec:solut}, we have shown that there exist a subsequence  ${(z_{\epsilon_i})}_{i\in\N}$ with $\epsilon_i\to 0$ and a continuous function $\bar z \colon [0,T]\to \R$ such that $z_{\epsilon_i}\to \bar z$ uniformly in $\CC([0,T])$. 
Setting $\bar \eta (t) =-\Phi'(\bar z(t))+\ell(t)$, we can apply Lemma \ref{lemma:dis} to the subsequence $z_{\epsilon_i}$ and get that 
\begin{equation*}
\liminf_{i\to\infty} \DD_{\epsilon_i}(z_{\epsilon_i})\leq \DD(\bar z)
\end{equation*}
We can therefore apply Proposition \ref{prop:conv} to find that $\bar z$ is a solution of \eqref{eq:din_Clim} for $\bar z(0)=z_0$ and so, as we have seen, of \eqref{eq:din_lim}. 

It is however well known in literature that problem \eqref{eq:din_lim} has only one solution for each choice of $z_0$ (cf.~\cite{Mie05,MieRou15}). This implies that actually the whole sequence $z_\epsilon$ converges to $\bar z$. Suppose by contradiction that there exists a subsequence $(z_{\epsilon_k})_{k\in \N}$ with $\epsilon_k\to 0$ such that $\norm{z_{\epsilon_k}-\bar z}_{\infty}>\bar \delta$, for some $\bar \delta>0$ and every $k\geq 0$. Then we can repeat the same reasoning done for $z_\epsilon$, to find a  function $\hat z\in \CC([0,T,])$, and a subsequence of $z_{\epsilon_k}$ that converges to $\hat z$. But, proceeding as above, $\hat z$ must be a solution of \eqref{eq:din_Clim} with $\hat z(0)=z_0$, and so, because of the uniqueness of the solutions, $\hat z=\bar z$, contradicting  $\norm{z_{\epsilon_k}-\bar z}_{\infty}>~\bar \delta$.

To complete the proof, it remains only to prove \eqref{eq:th_Rconv}. Let us first notice that, since \eqref{eq:din_eps} gives $\epsilon^\gamma \dert z_\epsilon(t)=\xi_\epsilon(t)$, a straightforward computation shows that $\Rc_\epsilon(\dert z_\epsilon(t))=\Rc_\epsilon^*(\xi_\epsilon(t))$ for almost every $t\in [0,T]$.  Moreover, since $\Rc_\epsilon(\dert z_\epsilon(t))+\Rc_\epsilon^*(\xi_\epsilon(t))=\dert z_\epsilon(t)\xi_\epsilon(t)$ for almost every $t$, by the chain rule we get, for every $0\leq t_1<t_2\leq T$,
\begin{equation*}
\int_{t_1}^{t_2}2\Rc_\epsilon(z_\epsilon(s))\dd s=\EE_\epsilon(t_2,z_\epsilon(t_2))-\EE_\epsilon(t_1,z_\epsilon(t_1))+\int_{t_1}^{t_2}\dert\ell(s)z_\epsilon(s)\dd s
\end{equation*}
On the other hand, for the limit system, since $-D_z\EE(t,\bar z(t))\in \Omega_0$, it follows that $\Rc^*(-D_z\EE(t,\bar z(t)))=0$ for almost every $t\in [0,T]$. Thus, again by the chain rule
\begin{equation*}
\int_{t_1}^{t_2}\Rc(\bar z(s))\dd s=\EE(t_2,\bar z(t_2))-\EE(t_1,\bar z(t_1))+\int_{t_1}^{t_2}\dert \ell(s)\bar z(s)\dd s
\end{equation*}
Since, for $\epsilon\to 0$, we have $z_\epsilon\to \bar z$ uniformly and $\EE_\epsilon(t,z_\epsilon(t))\to \EE(t,\bar z(t))$ for every $t\in[0,T]$,  it follows that
\begin{equation*}
\int_{t_1}^{t_2}2\Rc_\epsilon(\dert z_\epsilon(s))\dd s \to \int_{t_1}^{t_2}\Rc(\dert{\bar z}(s))\dd s
\qquad \text{for every $0\leq t_1<t_2\leq T$} 
\end{equation*}
and the proof is complete.

\section*{Acknowledgement}
This work has been supported by the ERC Advanced Grant 340685-MicroMotility.

%%%%%%%%%%%%%%%%%%%%%%%%%%%%%%%%%%%%%%%%%%%%%%%%%%%%%%%%%%%%%%%%%%%%%%%%%%%%%%%%%%%%5

%%%%%%%		BIBLIO

%%%%%%%%%%%%%%%%%%%%%%%%%%%%%%%%%%%%%%%%%%%%%%%%%%%%%%%%%%%%%%%%%%%%%%%%%%%%%%%%%%%%%

%%%%%%%%%%%%%%%%%%%%%%%%%%%%%%%%%%%%%%%%%%%%%%%%%%%%%%%%%%%%%%%%%%%%%%%%%%%%%%%%%%%%%%%%%

\begin{thebibliography}{XX}

\bibitem{aBS09} F. Al-Bender and J. Swevers. 
Characterization of friction force dynamics. 
\textit{IEEE Control Systems} 28 (2009), 64--81.

\bibitem{Arn07} V.I. Arnold,  V.V. Kozlov, and A.I. Neishtadt, 
\textit{Mathematical aspects of classical and celestial mechanics}. 
 Springer, Berlin (2007).
 
\bibitem{deW95} C. Canudas De Wit,  H. Olsson, K.J. Astrom and P. Lischinsky, 
A new model for control of systems with friction. 
\textit{IEEE TTrans. Autom. Control} 40 (1995), 419--425.

\bibitem{DriBer} B. Drincic and D. S. Bernstein, 
A sudden-release bristle model that exhibits hysteresis and stick-slip friction, in \textit{IEEE American Control Conference} (2011), 2456--2461. 

\bibitem{DeSGidNos15} A. DeSimone, P. Gidoni, and G. Noselli, G.,  
Liquid Crystal Elastomer Strips as Soft Crawlers. 
\textit{J. Mech. Phys. Solids} 85 (2015), 254--272.

\bibitem{DeSTat12} A. DeSimone and A. Tatone, 
Crawling motility through the analysis of model locomotors: two case studies.
\textit{Eur. Phys. J. E} 35:85 (2012).

\bibitem{GidDeS16} P. Gidoni and A. DeSimone,
Stasis domains and slip surfaces in the locomotion of a bio-inspired two-segment crawler.
\textit{Submitted for publication (2016).}

\bibitem{GidNosDeS14} P. Gidoni, G. Noselli and A. DeSimone,  
Crawling on directional surfaces.
\textit{Int. J. Non-Linear Mech.} 61 (2014), 65--73.


\bibitem{Han} M.J. Hancock., K. Sekeroglu, and M.C. Demirel, 
Bioinspired directional surfaces for adhesion, wetting, and transport. 
\textit{Adv. Funct. Mater.} 22  (2012), 2223--2234.

\bibitem{HaeFri} D.A. Haessig and B. Friedland,
On the modeling and simulation of friction.
 \textit{J. Dyn. Syst. Meas. Control.} 113 (1991),
354--362.

\bibitem{KwaShi04} M. Kwak and H. Shindo, 
Frictional force microscopic detection of frictional asymmetry and anisotropy at (10\={1}4) surface of calcite. 
\textit{Phys. Chem. Chem. Phys.} 6  (2004), 129--133.

\bibitem{Lil98} M. Liley, D. Gourdon, D. Stamou, U. Meseth, T. M. Fischer, C. Lautz, H. Stahlberg, H. Vogel, N. A. Burnham, and C. Duschl, 
Friction anisotropy and asymmetry of a compliant monolayer induced by a small molecular tilt. 
\textit{Science} 280 (1998), 273--275.

\bibitem{Mah04} L. Mahadevan, S. Daniel, and M.K. Chaudhury, 
Biomimetic ratcheting motion of a soft, slender, sessile gel.
 \textit{Proc. Natl. Acad. Sci. USA} 101 (2004), 23--26.

\bibitem{Men06} A. Menciassi, D. Accoto, S. Gorini, and P. Dario,
Development of a biomimetic miniature robotic crawler.
\textit{Auton. Robot} 21 (2006), 155--163.


\bibitem{Mie05} A. Mielke,  Evolution of rate-independent systems., in
\textit{Handbook of Differential Equations, evolutionary equations,}  
C. Dafermos and E. Feireisl, editors, Elsevier (2005). 

\bibitem{Mie12} A. Mielke, Emergence of rate-independent dissipation from viscous systems with wiggly energies. 
\textit{Contin. Mech. Thermodyn.} 24 (2012), 591--603. 
(doi:\,10.1007/s00161-011-0216-7)

\bibitem{Mie15} A. Mielke, 
Variational approaches and methods for dissipative material models with multiple scales, in 
\textit{Analysis and Computation of Microstructure in Finite Plasticity}. Springer, Berlin (2015),125--155. 

\bibitem{MieTru12} A. Mielke and L.Truskinovsky, 
From discrete visco-elasticity to continuum rate-independent plasticity: rigorous results.
\textit{Arch. Rational Mech. Anal. } 203 (2012), 577--619.

\bibitem{MieRou15} A. Mielke and T. Roubíček, 
 \textit{Rate-Independent Systems. Theory and Application.}
 Springer, Berlin (2015).
 
 \bibitem{Mil88} G.S.P. Miller, The motion dynamics of snakes and worms.
 \textit{ACM Siggraph Computer Graphics} 22 (1988), 169--173. 


\bibitem{NosDeS14} G. Noselli and A. DeSimone, 
A robotic crawler exploiting directional frictional  interactions: experiments, numerics, and derivation of a reduced model.
\textit{Proc. Roy. Soc. London A} 470 (2014).


\bibitem{Pop10} V. Popov, 
\textit{Contact mechanics and friction: physical principles and applications.}
Springer, Berlin (2010).

\bibitem{PopGre14} V.L. Popov and J.A.T. Gray, 
Prandtl-Tomlinson Model: A Simple Model Which Made History, in 
\textit{The History of Theoretical, Material and Computational Mechanics-Mathematics Meets Mechanics and Engineering.}
Springer, Berlin (2014), 153--168.

\bibitem{PugTru05} G. Puglisi and Lev Truskinovsky,
Thermodynamics of rate-independent plasticity.
\textit{J. Mech. Phys. Solids} 53 (2005), 655--679.

\bibitem{RocWet} R.T. Rockafellar and R.J.-B Wets, 
\textit{Variational Analysis}. 
Springer, Berlin (1998).

\bibitem{ZimZei07} K. Zimmermann and I. Zeidis,
Worm-like locomotion as a problem of nonlinear dynamics.
 \textit{J. Theoret. Appl. Mech.} 45 (2007), 179--187.

\end{thebibliography}
\end{document}